\documentclass[11pt, reqno,letterpaper]{amsart}
\usepackage{amsmath}
\usepackage{amsthm}
\usepackage{amssymb}
\usepackage{amsfonts}
\usepackage{graphicx}
\usepackage{epstopdf}
\usepackage{epstopdf}
\usepackage{epstopdf}
\usepackage{caption}
\usepackage[hyperfootnotes=false]{hyperref}

\newtheorem{lemma}{Lemma}[section]
\newtheorem{proposition}{Proposition}[section]
\newtheorem{remark}{Remark}[section]
\newtheorem{theorem}{Theorem}[section]
\newtheorem{definition}{Definition}[section]
\newtheorem{assumption}{Assumption}[section]

\newtheorem{example}{Example}[section]
\numberwithin{equation}{section}

\newcommand{\Real}{\mathbb R}

\newcommand{\expec}{\mathbb{E}}
\newcommand{\F}{\mathcal{F}}
\newcommand{\prob}{\mathcal{P}}
\newcommand{\borel}{\mathcal{B}}
\newcommand{\allprob}{\mathfrak{M}}
\newcommand{\Q}{\mathcal{Q}}
\newcommand{\QP}{\mathcal{Q}_{\phi}}
\newcommand{\dstra}{\mathcal{H}}
\newcommand{\stra}{\mathcal H\times \mathbb{R}^k}
\newcommand{\D}{\mathcal{D}}
\newcommand{\sratio}{\psi}

\begin{document}

\title[]{Quantile Hedging in a Semi-Static Market \\ with Model Uncertainty}

\author[]{Erhan Bayraktar}
\thanks{We thank Yuchong Zhang for her suggestions, which help correct mistakes in an earlier version. Erhan Bayraktar is partially supported by the National Science Foundation (DMS-1613170) and the Susan M. Smith Professorship.}
\address[Erhan Bayraktar]{Department of Mathematics, University of Michigan, 530 Church Street, Ann Arbor, MI 48109, USA}
\email{erhan@umich.edu}
\author[]{Gu Wang}
\address[Gu Wang]{Department of Mathematical Sciences, Worcester Polytechnic Institute, Worcester, MA 01609, USA}
\email{gwang2@wpi.edu}

\begin{abstract}
With model uncertainty characterized by a convex, possibly non-dominated set of probability measures, the agent minimizes the cost of hedging a path dependent contingent claim with given expected success ratio, in a discrete-time, semi-static market of stocks and options. Based on duality results which link quantile hedging to a randomized composite hypothesis test, an arbitrage-free discretization of the market is proposed as an approximation. The discretized market has a dominating measure, which guarantees the existence of the optimal hedging strategy and helps numerical calculation of the quantile hedging price. As the discretization becomes finer, the approximate quantile hedging price converges and the hedging strategy is asymptotically optimal in the original market. \end{abstract}

\keywords{Quantile hedging, model uncertainty, semi-static hedging,  Neyman-Pearson Lemma.}

\maketitle
\thispagestyle{empty}

\section{\textbf{Introduction}}\label{intro}
This paper considers the quantile hedging problem in which the agent is uncertain about the probability distribution of the payoff from investments in stocks and options. This situation of ``model uncertainty" arises naturally from the modeling of the financial market. For example, if the investor tries to estimate a model for stock price dynamics from historical data, statistical analysis gives confidence intervals of model parameters, and therefore infinitely many possible distributions of stocks prices. Thus a model described by a single probability measure bears the risk of misspecification. In this paper, we assume that the model uncertainty is characterized by a convex set $\prob$ of probability measures, which is not dominated by a reference measure, as for example, in the case of stock price with volatility uncertainty.

The quantile hedging problem was first discussed in \cite{FL99}: while superhedging a contingent claim is often very expensive, the agent can effectively lower the required initial capital for hedging, at a cost of a small probability of failure, or more generally, a loss in the expected success ratio (see Definition \ref{success ratio}). Quantile hedging and a closely related problem, maximizing the outperformance of the hedging portfolio, are then studied in \cite{SC99, BHS12, LSY13} among others. In particular, \cite{LSY13} assumes model uncertainty with a dominating measure, which guarantees the existence of a strategy with the optimal performance relative to the target contingent claim. Our paper is the first to discuss the quantile hedging problem in a model uncertainty setting without a dominating measure. 

As in \cite{BHR2001}, we assume that static positions in the options are held until the terminal date, while the stocks can be traded dynamically. Arbitrage and superhedging duality with model uncertainty or semi-static trading are studied by many researchers, see e.g. \cite{BHZ13,BN13,BHP2013,PRN2013,BZ15,BBKN2015,ABPS2016,BDG2016,BFM2015,CHO2016}. \cite{BN13} is particularly relevant, which proves the Fundamental Theorem of Asset Pricing and superhedging duality in our setting. Also related is the literature on conic finance and acceptability pricing (see e.g. \cite{MC2010,MS2011,BCIR2013}), in which acceptable contingent claims are characterized by a set of probability measures, while these measures are all dominated by a reference measure. 

In Section \ref{dual}, we derive dual representations of the quantile hedging price and a closely related quantity: the maximum expected success ratio with given initial capital. These results are based on the superhedging duality in \cite{BN13} and generalize the quantile hedging duality in \cite{FL99} to the case of static trading in options and model uncertainty. The dual representations link quantile hedging to a randomized composite hypothesis test, and indicate that the optimal quantile hedging strategy is the superhedging strategy for the contingent claim modified by the optimal test. The difficulty in our setting is that the set of hypothesis (probability measures) may be non-dominated. None of the extant results on composite hypothesis test applies and the existence of the optimal hedging strategy is not guaranteed (see e.g. \cite{HS73,CK01,LSY13,Gushchin2015} and Remark \ref{hypothesis}). Furthermore, the complex structure of the set of martingale measures inhibits the calculation of the quantile hedging price.

Section \ref{disc} contains the main result of this paper: an approximation of the quantile hedging price, which also guarantees the existence of an asymptotically optimal hedging strategy. The approximation is carried out via a discretization of the path space of stock price and the definition of model uncertainty on the discretized space, which are themselves interesting. The natural discretization used in \cite{BHZ13,Dolinsky14} may lead to arbitrage opportunities in the discretized market (see Example \ref{counter}), because we drop the key assumption in these papers that $\prob$ includes all probability measures on the path space. While the argument for the no arbitrage condition and superhedging is easier under this assumption, the quantile hedging problem becomes uninteresting, because every path is of probability one under some model, and quantile hedging price becomes linear in the target expected success ratio (see Example 1 in Section \ref{example}). To deal with a general convex set of models, we add extreme values of the stocks price at each time to the original path space, and construct a set of probability measures on the discretized space which assign arbitrarily small, positive probabilities to the added paths (see Definition \ref{ndef}). Theorem \ref{NA-n} shows that the discretized market satisfies the no arbitrage condition with semi-static trading in stocks and options, if it is sufficiently close to the original market. 

Since a dominating measure exists in the discretized model, the generalized Neyman-Pearson Lemma in \cite{LSY13,Gushchin2015} gives the solution to the associated randomized composite hypothesis test from the dual representations in Section \ref{dual}, and guarantees the existence of the optimal quantile hedging strategy in the discretized market. The agent can use this strategy as an approximation in the original market. Its performance can be quantified (Theorem \ref{convergence} and \ref{convergence-c}): to achieve the target expected success ratio in the original market, the agent can use the approximate strategy corresponding to a higher expected success ratio, which is arbitrarily close to the target, with some extra initial capital, which accounts for the discretization error. Furthermore, if the set of models in the original market includes sufficiently many probability measures that are finitely supported, which always allow price movements greater than or equal to a given (arbitrarily small) threshold (see definition \ref{JC}), then the quantile hedging price in the discretized market converges to that in the original market, and the approximate strategy is asymptotically optimal. The maximum expected success ratio and the quantile hedging price can be calculated numerically, by solving a nonlinear programming problem, as demonstrated in Section \ref{example}. 

The rest of the paper is organized as follows: Section \ref{model} sets up a semi-static market with model uncertainty, and defines the quantile hedging problem. Section \ref{dual} presents the dual representations of the quantile hedging price and the maximum expected success ratio. Section \ref{discrete} defines the discretized market and the corresponding model uncertainty, which are shown to satisfy the no arbitrage condition. Section \ref{discrete-q} solves the quantile hedging problem in the discretized market. Section \ref{approximate}  examines the performance of the approximate hedging strategy in the original market and the convergence of the approximate quantile hedging price. Section \ref{example} shows  examples of numerical calculation of the quantile hedging price. Some technical lemmas are in the Appendix.

\subsection{Notations} The following is a summary of the measure theoretical notations that are used frequently in the rest of the paper: given a topological space $\Omega$, let $\borel(\Omega)$ be its Borel $\sigma$-field, and $\allprob(\Omega)$ be the set of all probability measures on $\borel(\Omega)$.  If $\Omega$ is a Polish space (separable and completely metrizable topological space), $A\subset \Omega$ is analytic if it is the image of a Borel subset of another Polish space under a Borel-measurable mapping. A function $f:\Omega\rightarrow [-\infty,\infty]$ is called upper semianalytic if the inverse image of $(c,\infty]$ is analytic for every $c\in\Real$. The universal completion of $\borel(\Omega)$ is defined as $\displaystyle\cap_{P\in\allprob(\Omega)}\borel(\Omega)^P$, where $\borel(\Omega)^P$ is its $P$-completion. Note that from \cite[Chapter 7]{BS78} and \cite[Chapter 3, Appendix 2]{DY79}, any Borel set in $\Omega$ is analytic and any analytic set is measurable with respect to the universal completion of $\borel(\Omega)$, which is referred to as a universally measurable set. Furthermore, any Borel function is upper semianalytic and universally measurable. Finally, for any $P,Q\in\allprob(\Omega)$ and $\prob\subset\allprob(\Omega)$, write $Q\ll P$ if $Q$ is absolutely continuous with respect to $P$, $Q\sim P$ if they are equivalent, and $Q\lll \prob$ if $Q \ll P$ for some $P\in\prob$.

\section{\textbf{Model}}\label{model}
In this section, we set up the model of a discrete-time financial market with non-dominated model uncertainty, and define the quantile hedging problem. Assumptions (in particular, Assumptions \ref{NA} and \ref{option}) made in this section apply to the rest of the paper, without further notice.

\subsection{Market}
Consider the setup in \cite{BN13}: let $T \in \mathbb{N}$, and $\Omega_1\subset \mathbb R^d$ be a bounded Polish space. Let $\Omega_0$ be a singleton, and for $t\in \left\{1, 2, \dots,T\right\}$, $\Omega_t =\Omega_0\times \Omega_1^t$, where $\Omega^t_1$ is the Cartesian product of $\Omega_1$. Let $\F_t$ be the universal completion of $\borel(\Omega_t)$. Denote $(\Omega_T,\F_T)$ as $(\Omega,\F)$.

For each $t\in \left\{0,\cdots,T-1\right\}$, $\Omega_t$ represents the path space of $d$ stock prices up to $t$.  For each $\omega \in \Omega_t$, there is a non-empty convex set $\prob_t(\omega) \subseteq \allprob(\Omega_1)$, which represents the set of all possible models of the stock price at $t+1$, given the price history $\omega$. Assume that the graph of $\prob_t$ is analytic, which ensures that there exists a universally measurable selector (see \cite[Chapter 7]{BS78}): $P_t:\Omega_t \rightarrow \allprob(\Omega_1)$, such that $P_t(\omega;\cdot) \in \prob_t(\omega)$ for every $\omega \in \Omega_t$. Given kernels $P_t$ for $t=0,1,\dots,T-1$, a probability measure on $\Omega$ can be defined by: for any $A\in\mathcal{F}$,
\begin{equation}
P(A)=\int_{\Omega_1}\dots\int_{\Omega_1}I_{A}(\omega_0,\omega_1,\dots,\omega_T)P_{T-1}(\omega_1,\dots,\omega_{T-1};d\omega_T)\cdots P_0(\omega_0;d\omega_1),
\end{equation}
which is denoted as $P = P_0\otimes \cdots \otimes P_{T-1}$. The collection of all possible models of the market can be written as $\prob =\{P_0\otimes \cdots \otimes P_{T-1}: P_t(\cdot)\in \prob_t(\cdot), t = 0,1,\dots,T-1\}$, where each $P_t$ is a universally measurable selector of $\prob_t$. 

$A\in \mathcal F$ is $\prob$-polar if $P(A) = 0$ for every $P\in\prob$. A property holds $\prob$-quasi surely ($\prob$-q.s.), if it holds outside a $\prob$-polar set. For notational convenience, $\prob$ in front of q.s. is dropped in the rest of the paper, unless ambiguity arises. Also, $(\omega_0,\dots,\omega_t) \in \Omega_t$ is said to be $\prob$-polar, if $P\left((\omega_0,\dots,\omega_t)\times \left(\Omega_1\right)^{T-t}\right) =0$, for every $P\in\prob$.

Assume the risk-free rate is 0 in the financial market and the stock price at time $0\leq t\leq T$ is $S_t(\omega_0,\dots,\omega_t) = \omega_t$. Let $\dstra = \left\{H = (H_t)_{t=0}^{T-1}\right\}$ be the set of dynamic strategies in stocks, where for each $t$, $H_t:\Omega_t\rightarrow \mathbb R^d$ is $\mathcal F_t$-measurable, and $(H\cdot S)_t = \sum\limits_{i=0}^{t-1}H_i(S_{i+1}-S_i)$ is the return from strategy $H \in \dstra$ until time $t$. Let $\phi: \Omega\rightarrow \mathbb R^k$ be a $\mathcal{B}(\Omega)$-measurable random payoff of $k$ options that the investor can trade statically, i.e. buy or sell at time $0$ and hold the position until $T$, and $p\in\mathbb{R}^k$ be the price of $\phi$ at $t = 0$. $(H,q)\in\dstra\times \Real^k$ is called a semi-static strategy in stocks and options. With initial capital $x$, the corresponding terminal payoff is denoted by
$$G^{x,H,q}=x + (H\cdot S)_T + q(\phi-p).$$
For $H\in\dstra$ and $\omega = (\omega_1,\dots, \omega_T)\in\Omega$, let $||H(\omega)|| = \max\left(|H_0|, |H_1(\omega_0,\omega_1)|,\dots, |H_{T-1}(\omega_0,\dots,\omega_{T-1})|\right)$, where $|x| = \displaystyle\max_{1\leq i\leq d}|x_i|$ for any $x = (x_1,\dots, x_k)\in\Real^d$. 

Note that the boundedness of $\Omega_1$ means the agent assigns zero probability to prices beyond the bounds. This assumption is not very restrictive, because it is equivalent to saying that call options with strikes above the upper bounds and put options with strikes below the lower bounds are worthless to the agent. The common bounds for $d$ stocks also does not lose any generality. If the stocks have different bounds, we can always enlarge the original path space to a larger one with common bounds, and assign zero probability to the added paths.

\begin{definition}
For a $T$-period discrete time financial market with the path space of stock prices $\Omega$, model uncertainty $\prob\subset \allprob(\Omega)$, and options with payoff $\phi$ and price $p$, the no arbitrage condition with semi-static trading in stocks and options (denoted as NA$_{\phi}(\Omega,\prob)$) holds, if for any $(H,q)\in\dstra\times\Real^k$,
\begin{equation*}
G^{0,H,q} \geq 0 \text{ q.s. implies }    G^{0,H,q} = 0 \text{ q.s.}
\end{equation*}
Similarly, NA$(\Omega,\prob)$ denotes the no arbitrage condition with only dynamic trading in stocks ($q = 0$).
\end{definition}

Assume NA$_{\phi}(\Omega,\prob)$ holds, which is equivalent to (see \cite[Theorem 5.1]{BN13}):
\begin{assumption}\label{NA}
For every $P\in\prob$, there exists $Q\in\mathcal{Q}_{\phi}$, such that $P\ll Q $, where 
\begin{equation*}
\QP = \left\{Q\in\allprob\left(\Omega\right): Q\lll \prob, S \text{ is a martingale under } Q,\text{ and } \expec_Q[\phi] = p\right\}
\end{equation*}
is the set of martingale measures that fit with the given market price $p$ of $\phi$. 
\end{assumption}

The following superhedging duality, which is a part of Theorem 5.1 of \cite{BN13}, is useful for the discussion in Section \ref{dual}:
\begin{lemma}\label{suphedge}
	If Assumption \ref{NA} holds, and $X:\Omega\rightarrow \mathbb{R}$ is upper semianalytic, then the superhedging price
	\begin{equation*}
	\pi(X):=\inf\left\{x\in \mathbb{R}:\exists (H,q)\in \stra \text{ such that } G^{x,H,q}\geq X \text{ q.s.}\right\}
	\end{equation*}
	satisfies $\pi(X) = \sup\limits_{Q\in\mathcal{Q}_{\phi}}\mathbb{E}_Q[X]$, and there exists $(H,q)\in \stra$ such that $G^{\pi(X),H,q}\geq X$ q.s. Furthermore, the same holds with only dynamic trading in stocks, with $\QP$ replaced by $\Q = \left\{Q\in\allprob\left(\Omega\right): Q\lll \prob, S \text{ is a martingale under } Q\right\}$.
\end{lemma}

We also make the following assumption on $\phi$:
\begin{assumption}\label{option}
There does not exist $(H,q) \in \dstra\times\Real^k$ with $q\neq 0$, such that $G^{0,H,q} = 0$ q.s.
\end{assumption}

Assumption \ref{option} says that none of the options can be replicated by a semi-static strategy in stocks and other options, which does not lose any generality of the model. The reason is that under the no arbitrage condition, any terminal payoff from a position in a redundant option can be achieved by its replication portfolio with the same cost.

\subsection{Quantile Hedging}
An agent in this market trades with semi-static strategies and hedges against an upper semianalytic contingent claim $F: \Omega\rightarrow \mathbb{R}_+$ at $T$, where $\mathbb{R}_+ =\{x\in\mathbb R: x \geq 0\}$. The performance of a hedging strategy relative to $F$ is measured in terms of success ratio.
\begin{definition}(Success Ratio, cf. Definition (2.32) in \cite{FL99})\label{success ratio}
	For any upper semianalytic random payoff $F':\Omega\rightarrow\mathbb R_{+}$, the success ratio of $F'$ relative to $F$ is the universally measurable function defined as $\psi^{F'} := I_{\{F'\geq F\}} + \frac{F'}{F}I_{\{F'<F\}}\in [0,1]$.
\end{definition}

The goal of the agent is to minimize the hedging cost of $F$, with a given level of expected success ratio, in the worst case of all the models in $\prob$:
\begin{definition}(Quantile Hedging Price)\label{quant-price}
	For any $\alpha\in[0,1]$, the quantile hedging price of $F$ with expected success ratio $\alpha$ is:
	\begin{equation}
	\pi(\alpha,F) := \inf\left\{x\geq 0: \exists (H,q) \in \stra \text{ and } F'\in \mathcal{A}(\alpha) \text{ s.t. } G^{x,H,q} \geq F' q.s.\right\}, \label{goal1}
	\end{equation}
	where $\mathcal{A}(\alpha) = \left\{F':\Omega \rightarrow \mathbb{R}_{+} | F'\text{ is upper sermianalytic, and } \inf\limits_{P\in\prob}\mathbb{E}_P\left[\psi^{F'}\right] \geq  \alpha\right\}$.
\end{definition}

$\mathcal A(\alpha)$ is the set of upper semianalytic contingent claims which hedge $F$ with expected success ratio greater than or equal to $\alpha$, in the worst case of all the models in $\prob$ and the quantile hedging price is the minimum cost of superhedging some $F'\in \mathcal A(\alpha)$.

\begin{remark}
	We choose expected success ratio, instead of success probability, as the agent's criterion for the performance of the hedging portfolio. The reason is that, even without model uncertainty, the solution to the quantile hedging problem with given success probability may not exist, while it always exists if success ratio is considered (see \cite{FL99}). Furthermore, with model uncertainty, the agent has more flexibility when targeting an expected success ratio, because the hedging strategy's performance on the paths where the superhedging fails also counts. To target a success probability is more restrictive and can be very expensive. An example in which quantile hedging with any given success probability requires superhedging is provided in Section \ref{example}.
	\end{remark}

\section{\textbf{Dual Representations}}\label{dual}
In this section, as a preparation for the main results, we present dual representations of the quantile hedging price, and a closely related quantity: the maximum expected success ratio from given initial capital. The dual representations link quantile hedging to a randomized composite hypothesis test, which helps calculate of the quantile hedging price. These results are extensions of the dual representations in \cite{FL99} with $\prob$ being a singleton, and in \cite{LSY13} where $\prob$ and $\QP$ has a dominating measure.

The following proposition gives a dual representation of the quantile hedging price, of which the proof follows an idea similar to \cite{FL99}: quantile hedging a contingent claim is equivalent to superhedging a smaller one which achieves the target expected success ratio.
\begin{proposition}\label{quant-price1}
Let $\mathcal {A}^F(\alpha) = \{F'\in\mathcal A(\alpha): F'\leq F\}$. For $0\leq \alpha\leq 1$,
\begin{equation}
\pi(\alpha,F)= \inf_{F'\in\mathcal{A}^F(\alpha)}\sup_{Q\in\mathcal{Q}_{\phi}}\mathbb{E}_Q\left[F'\right].\label{dual1}
\end{equation}
\end{proposition}

\begin{proof}
For any $x\geq 0$, if there exists $(H,q)\in\stra$, $F'\in\mathcal A(\alpha)$ and $G^{x,H,q}\geq F'$ q.s., then from Lemma \ref{suphedge}, $x\geq\sup\limits_{Q\in\mathcal{Q}_{\phi}}\mathbb{E}_Q\left[F'\right]\geq \sup\limits_{Q\in\mathcal{Q}_{\phi}}\mathbb{E}_Q\left[F'\wedge F\right]$. 

Since $F'\wedge F$ is also upper semianalytic and $\psi^{F'\wedge F} = \psi^{F'}$, $F'\wedge F\in\mathcal A^F(\alpha)$. Thus $x\geq  \inf\limits_{F'\in\mathcal{A}^F(\alpha)}\sup\limits_{Q\in\mathcal{Q}_{\phi}}\mathbb{E}_Q\left[F'\right]$, which implies that $\pi(\alpha,F) \geq \inf\limits_{F'\in\mathcal{A}^F(\alpha)}\sup\limits_{Q\in\mathcal{Q}_{\phi}}\mathbb{E}_Q\left[F'\right]$.

On the other hand, if $F'\in\mathcal{A}^F(\alpha)$, let $x = \displaystyle\sup_{Q\in\mathcal{Q}_{\phi}}\mathbb{E}_Q\left[F'\right]\geq 0$. From Lemma \ref{suphedge}, there exists $(H,q)\in \stra$ such that $G^{x,H,q}\geq F'$ q.s. Since $F'\in\mathcal A^F(\alpha)\subset\mathcal A(\alpha)$, by definition $\displaystyle\sup_{Q\in\mathcal{Q}_{\phi}}\mathbb{E}_Q\left[F'\right] = x \geq \pi(\alpha,F)$. Since this holds for any $F'\in\mathcal{A}^F(\alpha)$, we obtain
$\displaystyle\inf_{F'\in\mathcal{A}^F(\alpha)}\sup_{Q\in\mathcal{Q}_{\phi}}\mathbb{E}_Q\left[F'\right] \geq \pi(\alpha,F)$.
\end{proof}

Together with Lemma \ref{suphedge}, this dual representation implies that if a minimizer $\hat F'$ in (\ref{dual1}) exists, then the optimal quantile hedging strategy is the superhedging strategy of the modified claim $\hat F'$.  

To better understand and help calculate the quantile hedging price, we also consider the ``inverse" (see Proposition \ref{inverse}) of the quantile hedging problem: with initial capital $x\geq 0$, the agent aims to maximize the expected success ratio among all non-negative, upper semianalytic payoffs, which are bounded by $F$, and can be superhedged from $x$, in the worst case of all the models in $\prob$:
\begin{equation}
V(x,F) = \sup_{F'\in\mathcal C(x)}\inf_{P\in\prob}\mathbb{E}_P\left[\psi^{F'}\right],\label{goal2}
\end{equation}
where $\mathcal C(x) = \{F': \Omega \rightarrow \mathbb{R}_{+}| F'\text{ is upper semianalytic, } F'\leq F, \text{ and } \pi(F')\leq x\}$. By definition, if there exists $\hat F' \in \mathcal{C}(x)$ that achieves the supremum in (\ref{goal2}), then $(\hat H,\hat q)$ that superhedges $\hat F'$ achieves the maximum expected success ratio. In particular, the no arbitrage condition implies that $C(0) = \{0\}$, and the supremum is always achieved.

\begin{lemma}\label{max-ratio} 
$V(x,F)$ is non-decreasing and concave in $x \in [0,\infty)$.
\end{lemma}

\begin{proof}
For any $x_1 > x_2 \geq 0$, $\mathcal C(x_2)\subset \mathcal C(x_1)$, thus $V(x_1,F) \geq V(x_2,F)$.

For concavity, consider $F'_1\in\mathcal C(x_1), F'_2\in\mathcal C(x_2)$, and $0\leq \lambda \leq 1$, let $F' = \lambda F'_1 + (1-\lambda)F'_2$. The superhedging prices satisfy
\begin{equation*}
\pi(F')\leq \pi(\lambda F'_1)  + \pi((1-\lambda)F'_2) \leq \lambda x_1 + (1-\lambda)x_2.
\end{equation*} 
Thus $F'\in \mathcal C(\lambda x_1 + (1-\lambda)x_2)$.

$F',F'_1$ and $F'_2 \leq F$ implies that $\psi^{F'} = \frac{F'}{F} =\lambda\frac{F'_1}{F} + (1-\lambda)\frac{F'_2}{F}= \lambda\psi^{F'_1} + (1-\lambda)\psi^{F'_2}$. Then
\begin{align*}
\lambda V(x_1,F) + (1-\lambda) V(x_2,F) =& \lambda\sup_{F'_1 \in \mathcal{C}(x_1)}\inf_{P\in\prob} \mathbb{E}_P\left[\psi^{F'_1}\right] + (1-\lambda)\sup_{F'_2 \in \mathcal{C}(x_2)}\inf_{P\in\prob} \mathbb{E}_P\left[\psi^{F'_2}\right]\\
=&\sup_{F'_i \in \mathcal{C}(x_i),i=1,2}\left(\inf_{P\in\prob} \mathbb{E}_P\left[\lambda\psi^{F'_1}\right]+\inf_{P\in\prob} \mathbb{E}_P\left[ (1-\lambda)\psi^{F'_2}\right]\right)\\
\leq&\sup_{F'_i \in \mathcal{C}(x_i),i=1,2}\inf_{P\in\prob} \mathbb{E}_P\left[\lambda\psi^{F'_1}+(1-\lambda)\psi^{F'_2}\right]\\
\leq& \sup_{F' \in \mathcal{C}(\lambda x_1+(1-\lambda)x_2)}\inf_{P\in\prob} \mathbb{E}_P\left[\psi^{F'}\right] = V(\lambda x_1 + (1-\lambda)x_2,F).\qedhere
\end{align*}
\end{proof}

\begin{proposition}\label{inverse}
If for every $x>0$, there exists $F'\in \mathcal C(x)$ that maximizes (\ref{goal2}), then $\pi(\alpha,F) = \inf\{x \geq 0: V(x,F) \geq \alpha\}$. Particularly, if $V(\pi(\alpha,F),F) \geq \alpha$, and $\hat F'$ maximizes (\ref{goal2}) with $x = \pi(\alpha,F)$, then the superhedging strategy for $\hat F'$ is the optimal quantile hedging strategy corresponding to the expected success ratio $\alpha$.
\end{proposition}

\begin{proof}
Let $\hat x = \inf\{x\geq 0: V(x,F) \geq \alpha\}$, then by the assumption, for any $\epsilon>0$, there exists $F'\in\mathcal C(\hat x +\epsilon)$, such that $\inf\limits_{P\in\prob}\mathbb{E}_P\left[\psi^{F'}\right] \geq \alpha$. Then $F'\in\mathcal A(\alpha)$, and can be superhedged with initial capital $\hat x + \epsilon$. Thus $\pi(\alpha,F)\leq \hat x + \epsilon$, which holds for every $\epsilon>0$. Therefore $\pi(\alpha,F)\leq \hat x$.

On the other hand, from the definition of the quantile hedging price, for any $\epsilon > 0$, there exists $(H^{\epsilon},q^{\epsilon})\in\stra$ and $F'\in\mathcal A(\alpha)$, such that $G^{\pi(\alpha,F)+\epsilon,H,q} \geq F' \geq F'\wedge F$ q.s. Since $\psi^{F} = \psi^{F'\wedge F}$, $F'\wedge F \in \mathcal A^F(\alpha)$. This implies that $V(\pi(\alpha,F) + \epsilon,F)\geq \alpha$, and $\pi(\alpha,F) + \epsilon \geq \hat x$, which holds for any $\epsilon >0$. Thus $\pi(\alpha,F) \geq \hat x$.

If $V(\pi(\alpha,F),F) \geq \alpha$ and $\hat F'$ is the maximizer in (\ref{goal2}), then $\inf\limits_{P\in\prob}\mathbb{E}_P\left[\psi^{\hat F'}\right] \geq \alpha$ and $\hat F' \leq F$, which implies that $\hat F'\in\mathcal A^F(\alpha)$. Then from Proposition \ref{quant-price1},  $\pi\left(\hat F'\right) \geq  \pi(\alpha,F)$. Furthermore, since $\hat F' \in \mathcal C(\pi(\alpha,F))$, from its definition we obtain $\sup\limits_{Q\in\mathcal{Q}_{\phi}}\mathbb{E}_Q\left[\hat F'\right] = \pi\left(\hat F'\right) \leq \pi(\alpha,F)$. Thus $\pi\left(\hat F'\right) = \pi(\alpha,F)$, and the superhedging strategy of $\hat F'$ is the optimal quantile hedging strategy.  
\end{proof} 

Notice that given $x\geq 0$, for any $F'\in\mathcal C(x)$, the definition of its success ratio implies that $F\psi^{F'} = F'$ and $\psi^{F'}\in[0,1]$. Letting $\mathcal C = \{F': \Omega \rightarrow \mathbb{R}_{+}| F'\text{ is upper semianalytic, and } F'\leq F\}$, then $F'\in\mathcal C(x)$ is equivalent to the constraint that $\psi:\Omega \rightarrow [0,1]$, $\pi(F\psi) \leq x$ and $\psi = \frac{F'}{F}$ for some $F'\in \mathcal C$. Thus (\ref{goal2}) can be written as
\begin{align}\label{test}
&{V}(x,F) = \sup_{\psi:\Omega \rightarrow [0,1]}\inf_{P\in\prob} \mathbb{E}_{P}[\psi]\\
&\text{subject to } \sup_{Q\in \QP}\mathbb{E}_{Q}[F\psi] \leq x, \text{ and } \psi = \frac{F'}{F} \text{ for some }F' \in \mathcal C.\nonumber
\end{align}
(\ref{test}) can be regarded as a randomized composite hypothesis test discussed in \cite{LSY13,Gushchin2015}, in which $\prob$ and $\QP$ (together with $F$) correspond to the set of alternate and null hypothesis, respectively, and the goal of (\ref{test}) is to find an optimal test, which maximizes the power of the test, given significance level $x$. There is also a constraint that only tests equal to ratios of two upper semianalytic functions are considered. The difference from \cite{LSY13,Gushchin2015} is that there does not exist a dominating measure for both $\prob$ and $\QP$.

To summarize, the dual representations in this section suggest that to solve the quantile hedging problem, we can try to solve the associated hypothesis testing problem. If $V(\pi(\alpha,F),F) \geq \alpha$, and the corresponding optimal test, or equivalently, the maximizer $\hat F'$ in (\ref{goal2}) exists, then the superhedging strategy of $\hat F'$ achieves expected success ratio $\alpha$ and thus is the optimal quantile hedging strategy.

\begin{remark}\label{hypothesis}
To the best of our knowledge, the extant generalization of the Neyman-Pearson Lemma, which guarantees the existence of a solution to composite hypothesis tests all require certain properties of the set of hypotheses (probability measures) under consideration. None of these assumptions hold in our setting, and the optimal test may not exist. \cite{CK01,LSY13,Gushchin2015} assume that a dominating measure exists. \cite{HS73,Augustin02} require that $\prob$ and $\QP$ are disjoint, $\mathcal P = \{P\in\allprob(\Omega): P\leq v_1\}$ and $\mathcal Q_{\phi} = \{Q\in\allprob(\Omega): Q\leq v_2\}$, where $v_1$ and $v_2$ are 2-alternating capacities, which does not always hold: (i) Given $\mathcal P$, if we define $v_1 (A) = \sup\limits_{P\in\mathcal P}P(A)$ for every $A\in\borel (\Omega)$, then $v_1$ is not necessarily a 2-alternating capacity (see Example 2 in [20]). (ii) Even if $v_1$ is a 2-alternating capacity, the set $\{P\in\allprob(\Omega): P\leq v_1\}$ may be strictly larger than $\mathcal P$ (see Example 1 in [20]). Furthermore, the dual representations in this section is not very helpful to the calculation of the quantile hedging price or the maximum expected success ratio, because $\QP$ is very difficult to characterize. 

Thus, in the next section, we provide an approximation method that helps the agent calculate an approximate quantile hedging price numerically, and guarantees the existence of an approximate hedging strategy, which is asymptotically optimal.
\end{remark}

\begin{remark}
All the results in this section hold if $F$ and $F'$ in Definitions \ref{success ratio}, \ref{quant-price} and (\ref{goal2}) are assumed to be Borel measurable\footnote{We thank one of the reviews to point this out.}. Every Borel measurable function is upper semianalytic, and we work with upper semianalytic functions, for the sake of generality.
\end{remark}

\section{\textbf{Approximation of the quantile hedging price and strategy}}\label{disc}
In this section and for the rest of the paper, without loss of generality, assume that the initial stock price is $S_0 = \mathbf{1}$, where for any constant $a$, $\mathbf{a}$ is the vector of the appropriate dimension, with entries identically equal to $a$, and $\Omega_1 = [a,b]^{d}$, where $0 \leq a  < 1<  b < \infty$ are dyadic numbers. 
 
\subsection{Discretized Market}\label{discrete}
In order to approximate the quantile hedging price and the hedging strategy, we discretize the path space, define model uncertainty in the discretized market and show that it satisfies the no arbitrage condition when the discretization is sufficiently fine.

We consider a discretized market, in which stock prices take values in  $\D_n = \{0,1/2^n,2/2^n,\dots\}$. For any $\omega^n = (\omega^n_0,\omega^n_1,\dots,\omega^n_T)\in \Omega_0 \times \left(\Omega_1\cap \D_n^d\right)^T$, denote as $\omega^n_{t,i}$ the $i$-th entry of $\omega^n_t$ for $0\leq t\leq T$, and we introduce the following notation:

\begin{definition}
Let $J^n_1(\omega^n_0,\omega^n_1) = \{\omega^n_0\}\times \prod\limits_{i=1}^d K^{n}(\omega^n_{0,i},\omega^{n}_{1,i})$, where
\begin{equation}\label{J}
K^{n}(x,y)= \begin{cases}
\left[y-1/2^{n+1}, y+1/2^{n+1}\right]\cap[a,b] & \text{ if } y = x,\\
\left(y-1/2^{n+1}, y+1/2^{n+1}\right]\cap[a,b], &\text{ if } y > x,\\
\left[y-1/2^{n+1}, y+1/2^{n+1}\right)\cap[a,b], &\text{ if } y<x.
\end{cases}
\end{equation}
For $t\geq 2$, let $J^n_t(\omega^n_0,\dots, \omega^n_t) = J^n_{t-1}(\omega^n_0,\dots,\omega^n_{t-1})\times \prod\limits_{i=1}^d K^{n}(\omega^{n}_{t-1,i},\omega^{n}_{t,i})$, and write $J^n_T(\omega^n_0,\omega^n_1,\dots,\omega^n_T)$ as $J^n(\omega^n)$.
\end{definition}

As defined above, $J^n(\omega^n)$ is the collection of all the paths in $\Omega$ that is within $1/2^{n+1}$ from $\omega^n$ in supnorm.

\begin{lemma}\label{partition}
For every $\omega\in\Omega$, there exists a unique $\omega^n\in \Omega_0 \times \left(\Omega_1\cap \D_n^d\right)^T$, such that $\omega\in J^n(\omega^n)$, i.e. $\left\{J^n(\omega^n):\omega^n\in  \Omega_0 \times \left(\Omega_1\cap \D_n^d\right)^T\right\}$ is a partition of $\Omega$.
\end{lemma}

\begin{proof}
For each $\omega\in\Omega$, the corresponding $\omega^n\in  \Omega_0 \times \left(\Omega_1\cap \D_n^d\right)^T$ can be found in an inductive way. First, $\omega^n_0 = \omega_0$. Suppose for $t \geq 1$, $\omega^n_0,\dots,\omega^n_{t-1}$ are uniquely determined, such that $(\omega_0,\dots,\omega_{t-1})\in J^n_{t-1}\left(\omega^n_0,\dots,\omega^n_{t-1}\right)$. Then $(\omega_0,\dots,\omega_{t})\in J^n_{t-1}\left(\omega^n_0,\dots,\omega^n_{t-1}\right)\times \Omega_1$.

For each $1\leq i\leq d$, since $\left\{K^n(\omega^n_{t-1,i},y): y\in [a,b]\cap \D_n\right\}$ is a partition of $[a,b]$, there exists a unique $y_i$ such that $\omega_{t,i} \in K^n(\omega^n_{t-1,i},y_i)$. Thus, with $\omega^n_t = (y_1,\dots,y_d)$, which is uniquely determined, $(\omega_0,\dots,\omega_{t})\in  J^n_{t-1}(\omega^n_0,\dots,\omega^n_{t-1})\times \prod\limits_{i=1}^d K^{n}(\omega^{n}_{t-1,i},y_i) = J^n_t(\omega^n_0,\dots,\omega^n_t)$.
\end{proof}

Since  $\left\{J^n(\omega^n):\omega^n\in  \Omega_0 \times \left(\Omega_1\cap \D_n^d\right)^T\right\}$ is a partition of $\Omega$, for the discretized path space $\Omega_0\times \left(\Omega_1\cap \D_n^{d}\right)^T$, a natural choice for the model uncertainty on it is that for every $P\in\prob$, define $P^n$ as $P^n(\omega^n) = P(J^n(\omega^n))$ for every $\omega^n\in\Omega_0\times \left(\Omega_1\cap \D_n^{d}\right)^T$, which is used in \cite{BHZ13,Dolinsky14}. However, this definition may lead to arbitrage opportunities in the discretized market and modifications are needed as the following example demonstrates.

\begin{example}\label{counter}
Suppose $d=1$, $T = 2$, $a = \frac{1}{2}$ and $b = \frac{3}{2}$ and there are no options. Assume $\prob = \{P\}$ is a singleton and $P_0(1;\cdot)$ is the uniform distribution on $\left[\frac{1}{2},\frac{3}{2}\right]$. For each $\omega_1 \in \left[1,\frac{3}{2}\right]$, the model $P_1(\omega_1;\cdot)$ is the uniform distribution on  $\left[1,\frac{3}{2}\right]$. On the other hand, for each $n\in\mathbb N$, if $1 - \frac{1}{2^n} \leq \omega_1 < 1 - \frac{1}{2^{n-1}}$, the model $P_1(\omega_1,\cdot)$ is the uniform distribution on $\left[1 - \frac{1}{2^n},\frac{3}{2}\right]$. 

For any dynamic strategy $H$, if $(H\cdot S)_2 \geq 0$ $P$-a.s., then consider 
\begin{align*}
A^+ =& \left\{\omega\in\Omega: \omega_1\in\left(1,\frac{3}{2}\right], \omega_2 \in \left[\frac 1 2,\omega_1\right] \text{ if } H_1(1,\omega_1) \geq 0, \text{ or } \omega_2 \in \left[\omega_1,\frac{3}{2}\right] \text{ if } H_1(1,\omega_1) \leq 0\right\},\\
A^- =& \left\{\omega\in\Omega: \omega_1\in\left[\frac{1}{2},1\right), \omega_2 \in \left[\frac 1 2,\omega_1\right] \text{ if } H_1(1,\omega_1) \geq 0, \text{ or } \omega_2 \in \left[\omega_1,\frac{3}{2}\right] \text{ if } H_1(1,\omega_1) \leq 0\right\}.
\end{align*}
$P\left(A^+\right) >0$ and $P\left(A^-\right) >0$. Furthermore, $H_1(\omega_2-\omega_1) \leq 0$ on each $\omega\in A^+\cup A^-$, which implies that $H_0(\omega_1-1) \geq 0$ $P$-a.s. on $A^+\cup A^-$. Since $\omega_1 -1 >0$ in $A^+$ and $\omega_1 -1 <0$ in $A^-$, $H_0 = 0$. Similarly,  since for a.s. every $\omega_1$, the probability of positive and negative price change are both positive, $H_1 = 0$ a.s. Thus NA$(\Omega,\prob)$ holds.

 However, for each $n\geq 2$, with $\Omega_0 \times \left(\Omega_1\cap \D_n\right)^2$ and $P^n$ defined above, $P^n_0(1;1) = \frac{1}{2^n}$. Furthermore, $P^n_1\left(1;\{\omega^n_2 \geq 1\}\right) = 1$ and $P_1^n(1;\{\omega^n_2 > 1\}) > 0$. Thus a dynamic strategy $H^n = \left(0,1_{\{1\}}(\omega^n_1)\right)$ is an arbitrage strategy and this arbitrage opportunity does not disappear as $n$ increases.   
\end{example}

\subsubsection{The Discretized Path Space and Model Uncertainty}\label{Pn}
Let $\underline a_t$ and $\bar b_t$, $t=1,\dots,T$ be constants such that $\underline a_T < \cdots < \underline a_1 <a - 1/2$ and $\bar b_T > \cdots > \bar b_1 > b + 1/2$. For each $1\leq t \leq T$, let $E_t = \{x\in\mathbb{R}^d: x_i\in\{\underline a_t,\bar b_t\}, i = 1,\dots,d\}$, and $|E_t| = 2^d$ be the size of $E_t$. The following is the definition of the discretized path space, by adding extreme values from $E_t$ to $\Omega_1\cap \D_n^{d}$ in each period. The corresponding probability measures on the discretized space assign (arbitrarily) small probabilities to paths which hits $E_t$ at $1\leq t \leq T$, so that (i) the price always moves both up and down with positive probability, which excludes arbitrage opportunities (see Theorem \ref{NA-n}), and (ii) the discretized market stays close to the original model (see Proposition \ref{Mbound}). 

\begin{definition}\label{ndef}
	(i) Discretized path space $\Omega^n$. $\Omega^n_t = \Omega_0\times\prod\limits_{s=1}^{t}\left(\left(\Omega_1\cap \mathcal D^d_n\right) \cup E_s\right)$, for $t = 1,\dots,T$. Assume $n$ is sufficiently large so that $a,b\in \mathcal D_n$. Write $\Omega^n_T$ as $\Omega^n$, which is the discretized path space until $T$, and for $1\leq t\leq T$, let
	\begin{equation}
	A^{n,t} = \left\{\omega^n\in\Omega^n: \omega^n_s\in E_s, \text{ for some }1\leq s\leq t\right\}. \label{An}
	\end{equation} 
	
	(ii) Model uncertainty $\prob^n$. Given $\lambda \geq 0$, for each $P\in\prob$, define the corresponding $P^n\in\allprob(\Omega^n)$ by a sequence of conditional probabilities: 
	For $0\leq t\leq T-1$, and $(\omega^n_0,\omega^n_1,\dots,\omega^n_t)\in \Omega^n_t$, if $\omega^n_s \in E_s$ for some $1\leq s\leq t$, or $P\left(J^n_{t}(\omega^n_0,\dots,\omega^n_{t})\times \Omega_1^{T-t}\right) = 0$, then let 
	\begin{equation*}
	P^n_t\left(\omega^n_0,\dots,\omega^n_t; \omega^n_{t+1}\right) = \frac{1_{\{E_{t+1}\}}(\omega^n_{t+1})}{|E_{t+1}|}. 
	\end{equation*}
	Otherwise, let
	\begin{equation*}
	P^n_t(\omega^n_0,\dots,\omega^n_t; \omega^n_{t+1}) =\begin{cases}
	\frac{P\left(J^n_{t+1}(\omega^n_0,\dots,\omega^n_{t+1})\times \Omega_1^{T-t-1}\right)}{\left(1+\lambda\right)P\left(J^n_{t}(\omega^n_0,\dots,\omega^n_{t})\times \Omega_1^{T-t}\right)}, & \text{ if } \omega^n_{t+1} \in \Omega_1\cap \D_n^d,\\
	\frac{\lambda/|E_{t+1}|}{1+\lambda}, & \text{ if } \omega^n_{t+1} \in E_{t+1}.
	\end{cases} 
	\end{equation*}
For each $\omega^n\in\Omega^n$, let $P^n(\omega^n) = P^{n}_{T-1}\left(\omega^n_0,\dots,\omega^n_{T-1};\omega^n_T\right)\cdots P^{n}_1\left(\omega^n_0,\omega^n_1;\omega^n_2\right)P^{n}_0(\omega^n_0;\omega^n_1)$. Furthermore, for any $A\in\borel\left(\Omega^n\right)$, let $P^n(A) = \sum\limits_{\omega^n\in A}P^n(\omega^n)$. Denote the collection of all $P^n$ constructed above (corresponding to each $P\in\prob$) as $\prob^{n,\lambda}$. 

(iii) The definition of $\phi$ and $F$ are extended to paths that hits $E_t$ between $1\leq t\leq T$, with\footnote{The results in the rest of the paper, as shown from their proof, does not depend on the choice of the values of $\phi$ and $F$ on $A^{n,T}$, as long as $F\geq 0$.} $\phi(\omega^n) = \mathbf{0}$ and $F(\omega^n) = 1$, if $\omega^n \in A^{n,T}$.
\end{definition} 

By the construction in Definition \ref{ndef}, if $\omega^n\in \Omega^n\setminus A^{n,T} = \Omega_0 \times \left(\Omega_1\cap \D_n^d\right)^T$, then each of $P^n\in\prob^{n,\lambda}$ and the corresponding $P\in\prob$, 
\begin{equation}
P^n(\omega^n) = \frac{1}{(1+\lambda)^T}P\left(J^n(\omega^n)\right), \label{PtoPn}
\end{equation} 
including the case where $P\left(J^n(\omega^n)\right) = 0$, because each $P^{n}_t(\omega^n_0,\dots,\omega^n_t;\omega^n_{t+1})$ is decreased by the factor $\frac{1}{1+\lambda}$ compared to $P$. On the other hand, if $\lambda >0$, $P^n$ assigns a positive value to $A^{n,T}$, which is the collection of paths which hit $E_t$ at $1\leq t\leq T$, and is the main difference between $\Omega$ and $\Omega^n$. The next proposition shows that $P^n\left(A^{n,T}\right)$ can be arbitrarily small, by decreasing $\lambda$.

\begin{proposition}\label{Mbound}
	For $A^{n,t}$ defined in (\ref{An}), $P^n\left(A^{n,t}\right)= 1-\frac{1}{\left(1+\lambda\right)^t}$ for every $P^n\in\prob^{n,\lambda}$ and $1\leq t\leq T$.
\end{proposition}

\begin{proof}
	We prove this proposition by induction on $t$. $A^{n,1} = E_1\times \prod\limits_{s=2}^{T}\left(\left(\Omega_1\cap \mathcal D^d_n\right) \cup E_s\right)$, and for every $P^n\in\prob^{n,\lambda}$, 
	\begin{align*}
	P^n\left(A^{n,1}\right) =& \sum\limits_{\omega^n \in A^{n,1}} P^{n}_{T-1}\left(\omega^n_0,\dots,\omega^n_{T-1};\omega^n_T\right)\cdots P^{n}_1\left(\omega^n_0,\omega^n_1;\omega^n_2\right)P^{n}_0(\omega^n_0;\omega^n_1)\\
	=& \sum\limits_{\omega^n_1 \in E_1} \frac{\lambda/|E_1|}{1+\lambda}\sum\limits_{\omega^n_t\in\left(\left(\Omega_1\cap \mathcal D^d_n\right) \cup E_t\right), 2\leq t\leq T} P^{n}_{T-1}\left(\omega^n_0,\dots,\omega^n_{T-1};\omega^n_{T}\right)\cdots P^{n}_1\left(\omega^n_0,\omega^n_1;\omega^n_2\right).
    \end{align*}
    Since $\omega^n_1 \in E_1$, each $P^{n}_{t-1}\left(\omega^n_0,\dots,\omega^n_{t-1};\omega^n_{t}\right)$ for $2\leq t\leq T$ in the second term above is $\frac{1_{\{E_{t}\}}(\omega^n_{t})}{|E_{t}|}$, and the sum is $1$. Therefore, $A^{n,1}=\sum\limits_{\omega^n_1 \in E_1} \frac{\lambda/|E_1|}{1+\lambda} = |E_1|\frac{\lambda/|E_1|}{1+ \lambda} = \frac{\lambda}{1+\lambda}$.

	Assume the proposition holds for $s\leq t-1$, so that $P^n\left(A^{n,t-1}\right)= 1-\frac{1}{\left(1+\lambda\right)^{t-1}}$ for every $P^n\in\prob^{n,\lambda}$. Then, since $A^{n,t}\setminus A^{n,t-1} = \{\omega^n\in\Omega^n: \omega^n_t\in E_t, \text{ and } \omega^n_{s}\notin E_s, 1\leq s\leq t-1 \}$,
	\begin{align*}
	P^n\left(A^{n,t}\setminus A^{n,t-1}\right) =&\sum\limits_{\omega^n \in A^{n,t}\setminus A^{n,t-1}} P^{n}_{t-1}\left(\omega^n_0,\dots,\omega^n_{t-1};\omega^n_t\right)\cdots P^{n}_1\left(\omega^n_0,\omega^n_1;\omega^n_2\right)P^{n}_0(\omega^n_0;\omega^n_1)\\
	=&\sum\limits_{\omega^n_s\notin E_s, 1\leq s\leq t-1} P^{n}_{t-2}\left(\omega^n_0,\dots,\omega^n_{t-2};\omega^n_{t-1}\right)\cdots P^{n}_0(\omega^n_0;\omega^n_1) \sum\limits_{\omega^n_t \in E_t} \frac{\lambda/|E_t|}{1+\lambda}\\
	&\sum\limits_{\omega^n_s\in\left(\left(\Omega_1\cap \mathcal D^d_n\right) \cup E_s\right), t+1\leq s\leq T} P^{n}_{T-1}\left(\omega^n_0,\dots,\omega^n_{T-1};\omega^n_{T}\right)\cdots P^{n}_t\left(\omega^n_0,\omega^n_t;\omega^n_{t+1}\right)\\
	=& \left(\frac{1}{(1+\lambda)^{t-1}}\sum\limits_{\omega^n_s\notin E_s, 1\leq s\leq t-1} P\left(J^n_{t-1}\left(\omega^n_0,\dots,\omega^n_{t-1}\right)\times\Omega_1^{T-t+1}\right)\right)\frac{\lambda}{1+\lambda}\\
	&\sum\limits_{\omega^n_s\in\left(\left(\Omega_1\cap \mathcal D^d_n\right) \cup E_s\right), t+1\leq s\leq T} P^{n}_{T-1}\left(\omega^n_0,\dots,\omega^n_{T-1};\omega^n_{T}\right)\cdots P^{n}_t\left(\omega^n_0,\dots,\omega^n_t;\omega^n_{t+1}\right),
	\end{align*}
	which follows from the definition of $P^n$. Following the same argument as for $A^{n,1}$, the last sum above equals $1$. Furthermore, $\left\{(\omega^n_0,\dots,\omega^n_{t-1}): \omega^n_s\notin E_s, 1\leq s\leq t-1\right\} = \Omega_0 \times \left(\Omega_1\cap \D_n^d\right)^{t-1}$, and Lemma \ref{partition} implies that  $\sum\limits_{\omega^n_s\notin E_s, 1\leq s\leq t-1} P\left(J^n_{t-1}\left(\omega^n_0,\dots,\omega^n_{t-1}\right)\times\Omega_1^{T-t+1}\right) = P(\Omega) = 1$, and therefore $P^n\left(A^{n,t}\setminus A^{n,t-1}\right) = \frac{\lambda}{(1+\lambda)^{t}}$. Finally, since $A^{n,t-1}\subset A^{n,t}$, we obtain $P^n(A^{n,t}) = P^n(A^{n,t-1}) + P^n(A^{n,t}\setminus A^{n,t-1}) = 1-\frac{1}{\left(1+\lambda\right)^t}$.
\end{proof}

The next proposition shows that $P^n$ defined above is indeed a probability measure on $\Omega^n$ and $\prob^{n,\lambda}$ is convex. Note that under each $P^n$, given $(\omega^n_0,\dots,\omega^n_t)\in\Omega^n_t$, the probability measure $P^n_t(\omega^n_0,\dots,\omega^n_t;\cdot)$ is supported on only finitely many $\omega^n_{t+1}$'s. Therefore, we call each $P^n$ a tree model. 
\begin{proposition}
As constructed in Definition \ref{ndef}, each $P^n\in\prob^{n,\lambda}$ is a probability measure on $\borel\left(\Omega^n\right)$, and $\prob^{n,\lambda}$ is convex.	
\end{proposition}

\begin{proof}
	For each $P^n\in\prob^{n,\lambda}$ to be a probability measure on $\Omega^n$, it suffices to show that $P^n(\Omega^n) = 1$. Since $\left( \Omega_0 \times \left(\Omega_1\cap \D_n^d\right)^T\right) \cap A^{n,T} = \emptyset$, and $\Omega^n = \left( \Omega_0 \times \left(\Omega_1\cap \D_n^d\right)^T\right) \cup A^{n,T}$, by definition,
	\begin{align*}
	P^n(\Omega^n) =& P^n\left(\Omega_0 \times \left(\Omega_1\cap \D_n^d\right)^T\right) + P^n(A^{n,T}) = P^n\left(\Omega_0 \times \left(\Omega_1\cap \D_n^d\right)^T\right) + 1-\frac{1}{\left(1+\lambda\right)^T},
	\end{align*}
	where the last equation follows from Proposition \ref{Mbound}. 
	
	On the other hand, Lemma \ref{partition} implies that  $\left\{J^n(\omega^n):\omega^n\in  \Omega_0 \times \left(\Omega_1\cap \D_n^d\right)^T\right\}$ is a partition of $\Omega$. Thus, with the corresponding $P\in\prob$, (\ref{PtoPn}) implies that
	\begin{align*}
	P^n\left(\Omega_0 \times \left(\Omega_1\cap \D_n^d\right)^T\right)  = \sum\limits_{\omega^n\in \Omega_0 \times \left(\Omega_1\cap \D_n^d\right)^T}\frac{1}{\left(1+\lambda\right)^T}P(J^n(\omega^n)) = \frac{1}{\left(1+\lambda\right)^T}P(\Omega) = \frac{1}{\left(1+\lambda\right)^T},
	\end{align*}
	which indicates that $P^n(\Omega^n) = 1$. 
	
	For the convexity of $\prob^{n,\lambda}$, suppose $P^{n,1},P^{n,2}\in\prob^{n,\lambda}$, which are defined via discretization of $P^1, P^2\in \prob$, respectively. For any $\alpha\in[0,1]$, since $\prob$ is convex, there exists $P^3\in\prob$, such that $P^3 = \alpha P^1 + (1-\alpha)P^2$, from which $P^{n,3} \in \prob^{n,\lambda}$ can be defined. It suffices to prove that for each $\omega^n\in\Omega^n$, $\alpha P^{n,1}(\omega^n) + (1-\alpha)P^{n,2}(\omega^n) = P^{n,3}(\omega^n).$
	
	If $\omega^n \in \Omega_0 \times \left(\Omega_1\cap \D_n^d\right)^T$, then (\ref{PtoPn}) implies that
	\begin{align*}
	P^{n,3}(\omega^n) =& \frac{P^3(J^n(\omega^n))}{\left(1+\lambda\right)^T} =  \frac{\left(\alpha P^1+(1-\alpha)P^2\right)(J^n(\omega^n))}{\left(1+\lambda\right)^T} =\alpha P^{1,n}(\omega^n) + (1-\alpha)P^{2,n}(\omega^n).
	\end{align*}
	
	If $\omega^n \in A^{n,T}$, then let $\tau = \min\{1\leq t\leq T: \omega^n_t \in E_t\}$, and
	\begin{align*}
	P^{n,3}(\omega^n) =& \left(\prod_{s=0}^{\tau -2} P^{n,3}_s(\omega^n_0,\dots,\omega^n_s;\omega^n_{s+1})\right) P^{n,3}_{\tau-1}(\omega^n_0,\dots,\omega^n_{\tau-1};\omega^n_{\tau})\prod_{s=\tau}^{T-1} P^{n,3}_s(\omega^n_0,\dots,\omega^n_s;\omega^n_{s+1})\\	
	=&\frac{1}{\left(1+\lambda\right)^{\tau -1}} P^3\left(J^n_{\tau -1}(\omega^n_0,\dots,\omega^n_{\tau -1}) \times \Omega_1^{T-\tau +1}\right) \frac{\lambda/|E_{\tau}|}{1+\lambda}\prod\limits_{s=\tau}^{T-1}\frac{1_{\{E_{s+1}\}}(\omega^n_{s+1})}{|E_{s+1}|}\nonumber\\ 
	=&  \frac{1}{\left(1+\lambda\right)^{\tau -1}}\left(\alpha P^1+(1-\alpha)P^2\right)(J^n_{\tau -1}(\omega^n_0,\dots,\omega^n_{\tau -1}) \times \Omega_1^{T-\tau +1}) \frac{\lambda/|E_{\tau}|}{1+\lambda}
	\prod\limits_{s=\tau}^{T-1}\frac{1_{\{E_{s+1}\}}(\omega^n_{s+1})}{|E_{s+1}|}\nonumber\\
	=&\alpha P^{1,n}(\omega^n) + (1-\alpha)P^{2,n}(\omega^n).\qedhere
	\end{align*}
\end{proof}

\subsubsection{The No Arbitrage Condition NA$_{\phi}\left(\Omega^n,\prob^{n,\lambda}\right)$.} Suppose in the discretized market the stock price $S$ is an $\Real^d$-valued process with $S_t(\omega^n_0,\dots,\omega^n_t) = \omega^n_t$ for every $(\omega^n_0,\dots,\omega^n_t)\in\Omega^n_t$, $t = 0,\dots,T$. Let $\mathcal{H}^n$ be the set of dynamic strategies in stocks $H = \left(H_t\right)_{t=0}^{T-1}$, where $H_t$ is a function defined on $\Omega^n_t$ for $t = 0,1,\dots, T-1$. The next theorem shows that the no arbitrage condition holds in the discretized market if it is sufficiently close to the original market.

\begin{theorem}\label{NA-n}
If $\phi$ is continuous on $\Omega$ and $\lambda >0$ in Definition \ref{ndef}, then NA$_{\phi}(\Omega^n,\prob^{n,\lambda})$ holds for $n$ sufficiently large.
\end{theorem}

\begin{proof}
Suppose there exists a sequence $\{n_i\}_{i=1}^{\infty}$, such that $n_i\rightarrow \infty$, and for every $i$, there exist $H^{n_i}\in\dstra^{n_i}$ and $q^{n_i}\in\Real^k$ such that $G^{0,H^{n_i},q^{n_i}} \geq 0$, $ \prob^{n_i,\lambda}$-q.s., and $P^{n_i}(G^{0,H^{n_i},q^{n_i}}>0) > 0$ for some $P^{n_i}\in\prob^{n_i,\lambda}$. Without loss of generality, assume if $(\omega^{n_i}_0,\dots, \omega^{n_i}_t)\subset\Omega^{n_i}_t$ is $\prob^{n_i,\lambda}$-polar, then $H^{n_i}_s(\omega^{n_i}_0,\dots, \omega^{n_i}_s) = 0$ for every $t\leq s\leq T-1$, because this modification does not change the fact that $(H^{n_i},q^{n_i})$ is an arbitrage strategy. 

We will show that the existence of $\left\{(H^{n_i},q^{n_i})\right\}_{i=1}^{\infty}$ leads to the existence of $(H,q)\in\dstra\times\Real^k$ such that $G^{0,H,q} = 0$ $\prob$-q.s. and $q\neq 0$, which contradicts Assumption \ref{option}, in two steps.

\textbf{Step 1}: The construction of $(H,q)$. Since $\lambda>0$, from Lemma \ref{NA-dynamic}, for every $i\geq 1$, NA$(\Omega^{n_i},\prob^{n_i,\lambda})$ holds. Then $q^{n_i} \neq 0$, because otherwise $(H^{n_i}\cdot S)_T \geq 0$ $\prob^{n_i,\lambda}$-q.s. implies that $(H^{n_i}\cdot S)_T = 0$ $\prob^{n_i,\lambda}$-q.s., which contradicts that $\left(H^{n_i},q^{n_i}\right)$ is an arbitrage strategy.

Normalize $(H^{n_i},q^{n_i})$ so that $|q^{n_i}| = 1$ for every $i$. Since $\phi$ is continuous on the compact $\Omega$, $|\phi| \leq  C$ on each $\Omega^{n_i}$ for some $C>0$, independent of $i$. Thus $(H^{n_i}\cdot S)_T > -D$ $\prob^{n_i,\lambda}$-q.s. for some $D>0$. Then, Lemma \ref{H-bound-n} implies that without loss of generality, we can assume $||H^{n_i}||\leq M$ for some $M>0$ independent of $i$. By a standard selection using a diagonalization argument (c.f. \cite[Section 9.6.1]{Resnick99}), there exist $q\in \Real^k$ and $H=(H_t)_{t=0}^{T-1}$, where $H_t$ is a function defined on $\Omega_0 \times \left(\Omega_1\cap \D^d\right)^T$, where $\D = \cup_{n\in\mathbb{N}}\D_n$, such that (up to a subsequence) $|q^{n_i}-q|\rightarrow 0$, and for every $\omega\in\Omega_0 \times \left(\Omega_1\cap \D^d\right)^T$, $||H^{n_i}(\omega)- H(\omega)||\rightarrow 0$. Thus $||H||\leq M$ and $|q|=1$.

With $|\phi|\leq C$, for every $\epsilon>0$, there exists $\delta>0$ such that $\left(Td(b-a) + Ck + |p|k\right)\delta < \epsilon$, and for any $\omega\in \Omega_0 \times \left(\Omega_1\cap \D^d\right)^T$, there exists $i$ sufficiently large, such that $\omega\in\Omega^{n_i}$, $||H^{n_i}(\omega)-H(\omega)||\leq \delta$ and $|q^{n_i} - q|\leq \delta$. Thus on $\omega$,
\begin{align}
&\left|(H\cdot S)_T + q(\phi-p) - (H^{n_i} \cdot S)_T - q^{n_i}(\phi-p)\right|\nonumber\\
 &\qquad \leq \left|\left(\left(H-H^{n_i}\right)\cdot S\right)_T\right| + |(q-q^{n_i})\phi| + |(q-q^{n_i})p|\leq Td(b-a)\delta + Ck\delta  + |p|k\delta  < \epsilon.\label{H-HJ}
\end{align}

Next, extend the domain of $H$ to $\Omega$. Let
\begin{align*}
L^{\Omega} =& \left\{\omega\in\Omega:\exists i\geq 1, \text{ and } \omega^{n_i}\in\Omega^{n_i}, \text{ s.t. } \omega\in J^{n_i}(\omega^{n_i}), \text{ and }\omega^{n_i} \text{ is } \prob^{n_i,\lambda}\text{-polar}\right\},\\
K^{\Omega}=& \Omega\setminus L^{\Omega}.
\end{align*}
Notice that if $\omega\in J^{n_i}(\omega^{n_i})$ for some $\omega^{n_i}\in\Omega^{n_i}$, then $\omega^{n_i} \in \Omega^{n_i}\setminus A^{n_i,T}$. For this $\omega^{n_i}$, (\ref{PtoPn}) implies that, if $\omega^{n_i}$ is $\prob^{n_i,\lambda}$-polar, then $J^{n_i}(\omega^{n_i})$ is $\prob$-polar. Thus for any $P\in\prob$,
\begin{align}
P\left(L^{\Omega}\right) \leq P\left(\cup_{i\geq 1}\cup_{\omega^{n_i}\in \Omega^{n_i}: \prob^{n_i,\lambda}\text{-polar}}J^{n_i}(\omega^{n_i})\right)=0. \label{L}
\end{align} 
Thus we focus on the definition of $H$ on $K^{\Omega}$. Since $||H||\leq M$ on $\Omega_0 \times \left(\Omega_1\cap \D^d\right)^T$,  $H$ can be defined on each $\omega\in K^{\Omega}$ following the inductive argument in \cite[Theorem 4.1]{BHZ13}, which guarantees the existence of $\left\{\omega^{n_i}\in\Omega^{n_i}: \omega^{n_i} \text{ is not } \prob^{n_i,\lambda}\text{- polar}\right\}_{i=1}^{\infty}$ and $H(\omega) = (H_t(\omega_0,\dots,\omega_t))_{t=0}^{T-1}$, such that (up to a subsequence) $|\omega-\omega^{n_i}|\rightarrow 0$, $||H(\omega)-H(\omega^{n_i})||\rightarrow 0$, as $i\rightarrow\infty$ and $||H(\omega)||\leq M$.

\textbf{Step 2}: $(H,q)$ leads to a contradiction to Assumption \ref{option}. With $|\phi|\leq C$ and $||H||\leq M$, for any $\epsilon>0$, there exist $\delta>0$ such that $ \left(T d (b-a)+ M T d + |q|k\right)\delta < \epsilon$. On the other hand, for any $\omega\in K^{\Omega}$, the continuity of $\phi$ and the construction of $H$ imply that, there exists a sufficiently large $i$, and $\omega^{n_i}\in \Omega^{n_i}$, such that $\max\left(|\omega-\omega^{n_i}|,|\phi(\omega)-\phi(\omega^{n_i})|,||H(\omega^{n_i})-H(\omega)||\right)\leq \delta$, and that  $\omega^{n_i}$ is not $\prob^{n_i,\lambda}$-polar. Therefore,
\begin{align*}
&|(H\cdot S)_T(\omega^{n_i}) + q(\phi(\omega^{n_i})-p) - (H\cdot S)_T(\omega) + q(\phi(\omega)-p)|\\
&\qquad \leq |\left((H(\omega^{n_i})-H(\omega))\cdot S(\omega^{n_i})\right)_T| + |\left(H(\omega)\cdot \left(S(\omega^{n_i}) - S(\omega)\right)\right)_T| + k|q||\phi\left(\omega^{n_i}\right)-\phi(\omega)|\\
&\qquad \leq Td(b-a)\delta + M T d\delta + k|q|\delta < \epsilon,
\end{align*}

Note that $i$ can be chosen to be sufficiently large so that (\ref{H-HJ}) also holds for $\omega^{n_i}$, because if $\omega^{n_i}\in \Omega^{n_i}$ and $i<j$, then $\omega^{n_i}\in \Omega^{n_j}$. Then, since $\omega^{n_i}$ is not $\prob^{n_i,\lambda}$-polar, and that $(H^{n,i},q^{n,i})$ is an arbitrage opportunity, $(H^{n_i} \cdot S)_T(\omega^{n_i}) - q^{n_i}(\phi(\omega^{n_i})-p)\geq 0$, and thus,
\begin{align*}
(H\cdot S)_T(\omega) + q(\phi(\omega)-p) \geq & (H \cdot S)_T(\omega^{n_i}) - q(\phi(\omega^{n_i})-p) -\epsilon\\
\geq& (H^{n_i} \cdot S)_T(\omega^{n_i}) - q^{n_i}(\phi(\omega^{n_i})-p)-2\epsilon\geq -2\epsilon.
\end{align*}
Since this holds for any $\epsilon>0$, $(H\cdot S)_T(\omega) + q(\phi(\omega)-p)\geq 0$. The same holds for every $\omega\in K^{\Omega}$ and thus for any $P\in\prob$,
\begin{equation*}
P\left((H\cdot S)_T(\omega) + q(\phi(\omega)-p)\geq 0\right) \geq P\left(K^{\Omega}\right) = 1-P\left(L^{\Omega}\right) =1,
\end{equation*}
where the last equation follows from (\ref{L}). Then since $\Omega$ satisfies NA$_{\phi}(\Omega,\prob)$, $(H\cdot S)_T + q(\phi-p) = 0$ $\prob$-q.s., with $|q|=1$. This contradicts Assumption \ref{option} and implies that NA$_{\phi}(\Omega^n,\prob^{n,\lambda})$ must hold for $n$ sufficiently large.
\end{proof}

Notice that though Assumption \ref{option} is not restrictive at all financially, as explained in Section 2.1, it is essential for the no arbitrage condition to hold in the discretized market. The reason is that under the no arbitrage condition in the original market, Assumption \ref{option} guarantees that for any strategy $(H,q)$ with $q\neq 0$, there exists $P\in\prob$, such that $P((H\cdot S)_T+q(\phi-p) < 0) >0$ and $P((H\cdot S)_T+q(\phi-p) > 0) >0$. If a sequence $\left\{(H^{n_i},q^{n_i})\in \dstra^{n_i}\times \mathbb{R}^k\right\}_{i=1}^{\infty}$ converges to $(H,q)$, then $(H^{n_i},q^{n_i})$ should share the same property under some $P^{n_i}\in\prob^{n_i,\lambda}$ for $i$ sufficiently large, and thus can not be an arbitrage strategy.

Theorem \ref{NA-n} holds if $\lambda>0$. $\lambda = 0$ corresponds to the construction at the beginning of Section 4.1, and as Example \ref{counter} shows, does not always satisfy the no arbitrage condition. From the proof Lemma \ref{NA-dynamic}, the key to the above no-arbitrage argument is that at every $0\leq t\leq T-1$ and for every $P^n\in\prob^{n,\lambda}$, $P^n_t$ always assigns $\frac{\lambda/|E_{t+1}|}{1+\lambda}$ or $\frac{1}{|E_{t+1}|}$ to each point in $E_{t+1}$, so that price can always move both up and down by $c$ (defined in (\ref{c})) with positive probability, which guarantees NA$(\Omega^n,\prob^{n,\lambda})$. Thus $\lambda >0$ can be arbitrarily small, but can not be zero. 

\subsection{Quantile Hedging in the Discretized Market}\label{discrete-q}
Let $n$ be sufficiently large such that NA$_{\phi}\left(\Omega^n,\prob^{n,\lambda}\right)$. We can discuss the quantile hedging problem in the discretized market, as an approximation of the original market:  
\begin{definition}
	 In the discretized market $\Omega^n$, with model uncertainty $\prob^{n,\lambda}$, for $\alpha\in[0,1]$, the quantile hedging price of $F$ with expected success ratio $\alpha$ is:
	\begin{equation}\label{quant-price-n}
	\pi^n(\alpha,F) := \inf\left\{x\geq 0: \exists (H,q) \in \dstra^n \times \mathbb{R}^k \text{ and } F'\in \mathcal{A}^n(\alpha) \text{ s.t. } G^{x,H,q} \geq F' q.s.\right\},
	\end{equation}
	where $\mathcal{A}^n(\alpha) = \left\{F': \Omega^n\rightarrow \mathbb{R}_{+}| \inf\limits_{P^n\in\prob^{n,\lambda}}\mathbb{E}_{P^n}\left[\psi^{F'}\right] \geq  \alpha\right\}$.
\end{definition}

Notice that since $\Omega^n$ has only finitely many paths, the restriction that $F'$ is upper semianalytic in the definition of $\pi(\alpha,F)$ is dropped for $\pi^n(\alpha,F)$. 

The quantile hedging problem is easier to solve in the discretized market than in the original one, because there exists a dominating measure for $\prob^{n,\lambda}$ and the associated martingale measures. Let $\bar{\Omega}^n = \{\omega^n\in\Omega^n: \omega^n \text{ is not }\prob^{n,\lambda}\text{-polar}\}$ and $|\bar\Omega^n|$ be the number of paths in $\bar\Omega^n$. Define a probability measure $\bar P^n$ as: for each $\omega^n\in\Omega^n$, 
\begin{equation}\label{barPn}
\bar P^n(\omega^n)=\begin{cases}
\frac{1}{|\bar\Omega^n|}, &\text{ if } \omega^n\in\bar\Omega^n,\\
0, &\text{ otherwise}.
\end{cases}
\end{equation}
Then $P^n\ll \bar P^n$ for every $P^n\in\prob^{n,\lambda}$ and NA$_{\phi}\left(\Omega^n,\prob^{n,\lambda}\right)$ is equivalent to NA$_{\phi}(\Omega^n,\{\bar P^n\})$. Thus, the set of martingale measures equivalent to $\bar P^n$, which also fit the market price of $\phi$,
\begin{equation*}
\mathcal{Q}^n_{\phi}=\{Q^n \in \allprob\left(\Omega^n\right): Q^n\sim \bar P^n, S\text{ is a martingale under }Q^n,  \text{ and } \mathbb{E}_{Q^n}[\phi] = p\}
\end{equation*}
is not empty. Furthermore, the superhedging price for a contingent claim $F:\Omega^n\rightarrow \mathbb{R}$ in $\prob^{n,\lambda}$-q.s. sense, or equivalently, in $\bar P^n$-a.s. sense, is $\pi^n(F) = \displaystyle\sup_{Q^n\in\QP^n}\mathbb{E}_{Q^n}[F]$.

The next lemma shows that only bounded strategies need to be considered when calculating the quantile hedging price in $\Omega^n$, and gives a dual representation of $\pi^n(\alpha,F)$ similar to (\ref{dual1}).

\begin{lemma}\label{price-Hq-bound}
	Let $\dstra^n_{M} = \{H\in\dstra^n: ||H||\leq M \}$, $K_M = \{q\in\Real^k: |q|\leq M\}$. There exists $M>0$, independent of $n$ and $\alpha$, such that
	\begin{equation*}
	\pi^n(\alpha,F)=\inf\left\{x\geq 0: \exists (H,q) \in \dstra^n_M\times K_M, \text{ and } F'\in \mathcal{A}^{n,F}(\alpha) \text{ s.t. } G^{x,H,q} \geq F' \text{ }\prob^{n,\lambda}\text{-}q.s.\right\},
	\end{equation*}
	where $\mathcal{A}^{n,F}(\alpha) = \left\{F'\in\mathcal A^n(\alpha): F'\leq F\right\}$. Furthermore,
	\begin{equation}
	\pi^n(\alpha,F)= \inf_{F'\in\mathcal{A}^{n,F}(\alpha)}\sup_{Q\in\QP^n}\mathbb{E}_Q\left[F'\right].\label{dual1n}
	\end{equation}	
\end{lemma}

\begin{proof}
For any $F'\in\mathcal A^n(\alpha)$, $\psi^{F'\wedge F} = I_{\{F'\wedge F\geq F\}} + \frac{F'\wedge F}{F}I_{\{F'\wedge F<F\}} = \psi^{F'}$, and therefore $F'\wedge F\in \mathcal A^{n,F}(\alpha)$. Thus for any $x\geq 0$, the $F'$ which achieves the target expected success ratio $\alpha$ (if exists)  can always be chosen from $\mathcal A^{n,F}(\alpha)\subset \mathcal A^n(\alpha)$, and $\mathcal A^n(\alpha)$ in (\ref{quant-price-n}) can be replaced by $\mathcal A^{n,F}(\alpha)$. 
	
	Since $\Omega^n$ has finitely many paths, $|F|$ and $|\phi|$ are bounded on $\Omega^n$. Thus $\pi^n(\alpha,F) \leq D = \displaystyle\sup_{\Omega^n} |F|<\infty$ which is independent of $n$ and $\alpha$, and
	\begin{equation*}
	\pi^n(\alpha,F)=\inf\left\{0\leq x\leq D: \exists (H,q) \in \dstra^n\times \mathbb{R}^k, \text{ and } F'\in \mathcal{A}^{n,F}(\alpha) \text{ s.t. } G^{x,H,q} \geq F' \text{ }\prob^{n,\lambda}\text{-}q.s.\right\}.
	\end{equation*}

	For the boundedness of hedging strategies, if for any $ (H,q) \in \dstra^n\times \mathbb{R}^k$, if $G^{x,H,q}$ superhedges $F'\in \mathcal A^{n,F}(\alpha)$ and $x\leq D$, then $(H\cdot S)_T + q(\phi-p) \geq -D$ $\prob^{n,\lambda}$-q.s. By Lemma \ref{H-bound-n} and Lemma 4.2 in \cite{BHZ13}, there exists $M>0$ such that only $(H,q) \in \dstra^n_M\times K_M$ needs to be considered. Finally (\ref{dual1n}) follows the same argument as for Proposition \ref{quant-price1}, and we skip the proof here.
\end{proof}

As Proposition \ref{inverse} suggests, we solve the quantile hedging problem by first analyzing the maximization of the expected success ratio, 
\begin{equation*}
V^n(x,F) = \sup_{F'\in\mathcal C^n(x)}\inf_{P^n\in\prob^{n,\lambda}}\mathbb{E}_{P^n}\left[\psi^{F'}\right],
\end{equation*}
where $\mathcal C^n(x) = \{F': \Omega^n \rightarrow \mathbb{R}_{+}|F'\leq F, \text{ and } \pi^n(F')\leq x\}$. Letting $\mathcal{Z}^n = \left\{Z^n = \frac{dP^n}{d\bar P^n}:P^n\in\prob^{n,\lambda}\right\}$ and $\mathcal{G}^n = \left\{G^n = \frac{dQ^n}{d\bar P^n}:Q^n\in\mathcal{Q}^n_{\phi}\right\}$, then the corresponding hypothesis test can be written as
\begin{align}
&{V}^n(x,F) = \sup_{\psi: \Omega^n\rightarrow [0,1]}\inf_{Z^n\in\mathcal{Z}^n} \mathbb{E}_{\bar P^n}[Z^n\sratio]\label{test-n}\\
&\text{subject to } \sup_{G^n\in \mathcal{G}^n}\mathbb{E}_{\bar P^n}[G^nF\sratio] \leq x.\nonumber
\end{align}
Notice that since there is only finitely many paths in $\Omega^n$, the constraint that  $\psi = \frac{F'}{F}$, where $F'$ is upper semianalytic and $0\leq F'\leq F$ is removed, because every $\psi: \Omega^n\rightarrow  [0,1]$ satisfies this condition. Thus in the discretized market, quantile hedging can be solved by a randomized composite hypothesis test, in which both set of hypothesis are dominated.

\begin{proposition}\label{inverse-n}
In the discretized market $\Omega^n$ with model uncertainty $\prob^{n,\lambda}$,

(i) For every $x\geq 0$, the optimal test $\hat \psi^n$ for (\ref{test-n}) exists, and 
\begin{align}
V^n(x,F)=&  \inf_{a\geq 0}\left\{xa + \inf_{\mathcal{Z}^n\times \mathcal{G}^n}\mathbb{E}_{\bar P^n}\left[(Z^n-aG^nF)^+\right]\right\}\nonumber\\
 =& \inf_{a\geq 0}\left\{xa + \inf_{\prob^{n,\lambda}\times\QP^n}\sum\limits_{w^n\in \Omega^n}\left(P^n(w^n)-a Q^n(w^n)F(w^n)\right)^{+}\right\},\label{optimal-test}
\end{align}
which is non-decreasing, concave and continuous in $x \in [0,\infty)$.

(ii) The optimal quantile hedging strategy  $\left(\hat H^n,\hat q^n\right)$ exists, i.e. $G^{\pi^n(\alpha,F),\hat H^n,\hat q^n} \geq F'$, for some $F'\in \mathcal A^{n,F}(\alpha)$.
\end{proposition}

\begin{proof}
With a dominating measure $\bar P^n$, let $\bar {\mathcal Z}^n$ and $\bar {\mathcal G}^n$ be the closure of $\mathcal{Z}^n$ and $\mathcal {G}^n$ respectively, in $\bar P^n$-convergence. (\ref{optimal-test}) follows from \cite[Theorem 2.3]{LSY13} and \cite[Theorem 1.1]{Gushchin2015}, which also imply that for any $x\geq 0$, there exist $\hat{Z}^n \in \bar{\mathcal{Z}}^n$, $\hat{G}^n \in \bar{\mathcal{G}}^n$ and $\hat{\psi}^n:\Omega^n\rightarrow [0,1]$, such that 
\begin{equation}
\mathbb{E}_{\bar P^n}\left[\hat G^nF\hat{\psi}^n\right] = \sup_{G^n\in \mathcal{G}^n}\mathbb{E}_{\bar P^n}\left[G^nF\hat\psi^n\right]  = x, \text{ and } 
V^n(x,F) = \mathbb{E}_{\bar P^n}\left[\hat{Z}^n\hat{\psi}^n\right] = \inf_{Z^n\in\mathcal{Z}^n} \mathbb{E}_{\bar P^n}\left[Z^n\hat\psi^n\right].\label{LSY}
\end{equation} 
The monotonicity and concavity follows the same proof as for Lemma \ref{max-ratio}. The concavity implies that $V^n(x,F)$ is continuous on $(0,\infty)$. The fact that (\ref{optimal-test}) is upper semicontinuous in $x$ together with the monotonicity implies it is also continuous at $x=0$ (c.f. \cite[Proposition 1.2]{Gushchin2015}).

For part (ii), since for every $x\geq 0$, the optimal test, or equivalently the contingent claim that maximizes the expected success ratio always exists, together with the continuity from part (i), Proposition \ref{inverse} implies that $\pi^n(\alpha,F) = \min\{x \geq 0: V(x,F) \geq \alpha\}$ and $V^n\left(\pi^n(\alpha,F),F\right) \geq \alpha$.

Furthermore, let $\tilde{\psi}^n$ be the optimal test corresponding to $x = \pi^n(\alpha,F)$, then with $F' = F\tilde \psi^n$, (\ref{LSY}) implies that (i) $\inf\limits_{P\in\prob}\mathbb{E}_P\left[\psi^{F'}\right] \geq \alpha$ and thus $F'\in\mathcal A^{n,F}(\alpha)$, and (ii) $\sup\limits_{Q\in\mathcal{Q}_{\phi}}\mathbb{E}_Q\left[F'\right] =\pi(F') = \pi(\alpha,F)$. Thus, $\hat F'$ is a minimizer in (\ref{dual1n}), and the superhedging strategy for $F'$ is the optimal quantile hedging strategy, of which the existence is implied by Lemma \ref{suphedge}.  
\end{proof}  
 
In addition to the existence of the quantile hedging strategy, another nice property of the discretized market lies in the calculation of the quantile hedging price. Since there are finitely many paths in $\Omega^n$, (\ref{optimal-test}) is a nonlinear programming problem with $Z^n$, $G^n$ (which can be regarded as vectors, and are much easier to characterize than the martingale measures in $\Omega$) and $a$ as variables and the constraints that $Z^n\in \mathcal Z^n$, $G^n\in\mathcal G^n$ (they are Radon-Nikodym derivative of $P^n \in \prob^{n,\lambda}$ and $Q^n \in \QP^n$, respectively, with respect to $\bar P^n$) and $a\geq 0$. The maximum expected success ratio for any initial capital can the calculated numerically, as demonstrated in Section \ref{example}.

\begin{remark} The key of the approximation method is that it transforms the original problem into a setting with a dominating probability measure. It is a potentially useful tool in solving other problems in Mathematical Finance and Operations Research in the setting of \cite{BN13}, for example, approximating the optimal terminal payoff in utility maximization as discussed in \cite{Schied04,Schied05,Gundel05,Nutz2016}, by applying the techniques in \cite{Schied05}, which relies on the existence of a dominating measure.
\end{remark}
\subsection{The Approximate Quantile Hedging Price and Hedging Strategy}\label{approximate}
Let $(\hat H^n,\hat q^n)$ be the optimal quantile hedging strategy in $\Omega^n$. The definition of $\hat H^n$ can be extended to $\Omega$ in a piecewise constant way: for any $\omega\in\Omega$, let
\begin{equation}
\hat H^n(\omega) = \hat H^n(\omega^n), \text{ if } \omega\in J^n(\omega^n)  \text{ for some }\omega^n \in \Omega^n.\label{Hextension}
\end{equation}
Since Lemma \ref{partition} shows that $\left\{J^n(\omega^n):\omega^n\in  \Omega_0 \times \left(\Omega_1\cap \D_n^d\right)^T\right\}$ is a partition of $\Omega$, $H^n$ is well-defined on $\Omega$.

With no clue of what the optimal hedging strategy in the original market is, the agent can use $(\hat H^n,\hat q^n)$ as an approximate. The question is: how does this approximate strategy work? In particular, in order to achieve the target expected success ratio $\alpha$, what expected success ratio should be used to derive $(\hat H^n,\hat q^n)$ in the discretized market and what is the cost of this strategy? The next theorem shows that the agent needs to target a  higher (but arbitrarily close to $\alpha$) expected success ratio and prepare some extra initial capital, than what is originally optimal, which accounts for the discretization error, and decreases to $0$ as $n$ increases to infinity.

\begin{theorem}\label{convergence}
If $\phi$ and $F$ are Lipschitz continuous on $\Omega$, and NA$_{\phi}(\Omega^n,\prob^{n,\lambda})$ holds, then there exists a constant $M>0$, independent of $n$ and $\alpha$, such that,
\begin{equation*}
\pi(\alpha,F)\leq \pi^n(\alpha',F) + M/2^n.
\end{equation*}
where $\alpha' = \alpha + \left(1-\frac{1}{(1+\lambda)^T}\right)(1-\alpha)$. Furthermore, starting from $\pi^n(\alpha',F) + M/2^n$, $(\hat H^n, \hat q^n)$, the optimal quantile hedging strategy for $F$ in the discretized market with expected success ratio $\alpha'$, achieves the target $\alpha$ in the original market, after extending the definition of $\hat H^n$ to $\Omega$, as in (\ref{Hextension}).
\end{theorem}

\begin{proof}
Let $B^{n,T} = \bar\Omega^n\setminus A^{n,T}$, and $B = \left\{\omega\in\Omega: \omega\in J^n(\omega^n), \text{ for some } \omega^n \in B^{n,T}\right\}$, where $A^{n,T}$ is defined in (\ref{An}) and $\bar{\Omega}^n$ is the collection of all non-$\prob^{n,\lambda}$-polar paths in $\Omega^n$. We focus the following discussion on the pair of $\omega\in B$ and the corresponding $\omega^n$ such that $\omega\in J^n(\omega^n)$, because if $\omega^n\in A^{n,T}$, then $J^n(\omega^n) \cap \Omega = \emptyset$, and  if $\omega^n$ is $\prob^{n,\lambda}$-polar, then (\ref{PtoPn}) implies that $J^n(\omega^n)$ is $\prob$-polar.

Lemma \ref{price-Hq-bound} implies that, without loss of generality, we can assume $\max\left(\left|\hat q^n\right|,||\hat H^n||\right)\leq C_1$, for a constant $C_1>0$ independent of $n$ and $\alpha'$. The Lipschitz continuity of $F$ and $\phi$ implies that $|\phi(\omega)-\phi(\omega^n)|\leq C_2/2^{n}$ and $|F(\omega)-F(\omega^n)|\leq C_3/2^n$, for constants $C_2>0$ and $C_3>0$.

Since $(\hat H^n,\hat q^n)$ is the optimal quantile hedging strategy with expected success ratio $\alpha'$ in $\Omega^n$, there exists $F^{'n}\in \mathcal A^{n,F}(\alpha')$, such that $G^{\pi^n(\alpha',F),\hat H^n,\hat q^n} \geq F^{'n}$ $\prob^{n,\lambda}$-q.s. Define an $\mathcal F$-measurable upper semianalytic function $F':\Omega \rightarrow \mathbb{R}_+$ by $F'(\omega) = F^{'n}(\omega^n) + C_3/2^n$ if $\omega \in J^n(\omega^n)$, and let $M = C_1\left(T d+C_2 k\right)+ C_3$. The rest of the proof shows that (i) $G^{\pi^n(\alpha',F)+M/2^n,\hat H^n,\hat q^n} \geq F'$ $\prob$-q.s. (after extending the definition of $\hat H^n$ to $\Omega$ as in (\ref{Hextension})) and (ii) $F'\in \mathcal A(\alpha)$.

(i) With any initial capital $x\geq 0$, since $\omega\in J^n(\omega^n)$,
\begin{align*}
&\left|G^{x,\hat H^n,\hat q^n}(\omega) - G^{x,\hat H^n,\hat q^n}(\omega^n)\right| =\left|\left(\hat H^n(\omega^n)\cdot\left(S(\omega)-S(\omega^n)\right)\right)_T + \hat q^n(\phi(\omega)-\phi(\omega^n))\right|\\
&\leq C_1T d/2^n + C_1C_2k/2^n=C_1\left(T d+C_2 k\right)/2^n.
\end{align*}
Thus for quasi-surely every $\omega\in\Omega$ and the corresponding $\omega^n \in \Omega^n$ such that $\omega\in J^n(\omega^n)$,
\begin{align*}
G^{\pi^n(\alpha',F)+M/2^n,\hat H^n,\hat q^n}(\omega) &\geq G^{\pi^n(\alpha',F)+M/2^n,\hat H^n,\hat q^n}(\omega^n) - C_1\left(T d+C_2 k\right)/2^n\\
 &= G^{\pi^n(\alpha',F),\hat H^n,\hat q^n}(\omega^n) + C_3/2^n \geq F^{'n}(\omega^n) + C_3/2^n = F'(\omega).
\end{align*}

(ii) In terms of success ratio of $F'$ relative to $F$, if $F^{'n}(\omega^n) \geq F(\omega^n)$, then
\begin{equation}
F'(\omega) = F^{'n}(\omega^n) + C_3/2^n \geq F(\omega^n) + C_3/2^n \geq F(\omega).\label{ineq1}
\end{equation}
On the other hand, if $F'(\omega)< F(\omega)$, then
\begin{equation}
F^{'n}(\omega^n) = F'(\omega) - C_3/2^n < F(\omega) - C_3/2^n \leq F(\omega^n).\label{ineq2}
\end{equation}
Furthermore, if $F^{'n}(\omega^n)< F(\omega^n)$, then
\begin{align}
&\frac{F'(\omega)}{F(\omega)} = \frac{F^{'n}(\omega^n) + C_3/2^n }{F(\omega)} \geq \frac{F^{'n}(\omega^n) + C_3/2^n }{F(\omega^n)+ C_3/2^n}\geq \frac{F^{'n}(\omega^n) }{F(\omega^n)}. \label{ineq3}
\end{align}
(\ref{ineq1}) and (\ref{ineq2}) imply that the success ratio of $F'$ satisfies
\begin{align*}
&\psi^{F'}(\omega) = I_{\{F'(\omega)\geq F(\omega)\}} + \frac{F'(\omega)}{F(\omega)}I_{\{F'(\omega)< F(\omega)\}}\nonumber\\
=&I_{\{F'(\omega)\geq F(\omega),F^{'n}(\omega^n)\geq F(\omega^n)\}} + I_{\{F'(\omega)\geq F(s),F^{'n}(\omega^n)< F(\omega^n)\}} + \frac{F'(\omega)}{F(\omega)}I_{\{F'(\omega)< F(\omega),F^{'n}(\omega^n)< F(\omega^n)\}}\nonumber\\
\geq&I_{\{F^{'n}(\omega^n)\geq F(\omega^n)\}} + \frac{F^{'n}(\omega^n)}{F(\omega^n)}I_{\{F'(\omega)\geq F(\omega),F^{'n}(\omega^n)< F(\omega^n)\}} + \frac{F'(\omega)}{F(\omega)}I_{\{F'(\omega)< F(\omega),F^{'n}(\omega^n)< F(\omega^n)\}}.\nonumber
\end{align*}
From (\ref{ineq3}), the above is greater than or equal to
\begin{align*}
&I_{\{F^{'n}(\omega^n)\geq F(\omega^n)\}} + \frac{F^{'n}(\omega^n)}{F(\omega^n)}I_{\{F'(\omega)\geq F(\omega),F^{'n}(\omega^n)< F(\omega^n)\}} + \frac{F^{'n}(\omega^n)}{F(\omega^n)}I_{\{F'(\omega)< F(\omega),F^{'n}(\omega^n)< F(\omega^n)\}}\\
=&I_{\{F^{'n}(\omega^n)\geq F(\omega^n)\}} +\frac{F^{'n}(\omega^n)}{F(\omega^n)}I_{\{F^{'n}(\omega^n)< F(\omega^n)\}} = \psi^{F^{'n}}(\omega^n).
\end{align*}

Recall that for every $\omega^n\notin \bar\Omega^n = A^{n,T}\cup B^{n,T}$, $J^n(\omega^n)$ is $\prob$-polar, and (\ref{PtoPn}) implies that, for every $\omega^n\in B^{n,T}$, $P^n\in\prob^{n,\lambda}$ and the corresponding $P\in\prob$, $P^n(\omega^n)(1+\lambda)^T = P(J^n(\omega^n))$. Thus
\begin{align}
\inf_{P\in\prob}\mathbb{E}_P\left[\psi^{F'}\right] \geq & \inf_{P\in\prob}\mathbb{E}_P\left[I_{B}\psi^{F'}\right] + \inf_{P\in\prob}\mathbb{E}_P\left[I_{\Omega\setminus B}\psi^{F'}\right] = \inf_{P\in\prob}\mathbb{E}_P\left[I_{B}\psi^{F'}\right]\nonumber\\
=& \inf_{P\in\prob}\sum\limits_{\omega^n \in B^{n,T}}\mathbb{E}_P\left[\psi^{F'}|J^n(\omega^n)\right]P\left(J^n(\omega^n)\right) \geq \inf_{P\in\prob}\sum\limits_{\omega^n \in B^{n,T}}\psi^{F^{'n}}(\omega^n)P\left(J^n(\omega^n)\right)\nonumber\\
= & \inf_{P^n\in\prob^{n,\lambda}}\sum\limits_{\omega^n\in B^{n,T}}\psi^{F^{'n}}(\omega^n)P^n(\omega^n)(1+\lambda)^T = (1+\lambda)^T\inf_{P^n\in\prob^{n,\lambda}}\mathbb{E}_{P^n}\left[I_{B^{n,T}}\psi^{F^{'n}}\right]\nonumber\\
\geq&(1+\lambda)^T \inf_{P^n\in\prob^{n,\lambda}}\left(\mathbb{E}_{P^n}\left[\psi^{F^{'n}}\right]- P^n\left(A^{n,T}\right)\right)\nonumber\\ 
\geq& (1+\lambda)^T \alpha'- (1+\lambda)^T+1  = \alpha,\label{alphap}
\end{align}
which follows from that $P^n(A^{n,T}) = 1 -\frac{1}{(1+\lambda)^T}$, as shown in Proposition \ref{Mbound}, and thus $F'\in \mathcal A(\alpha)$. (\ref{alphap}) also indicates that $x_n = \pi^n(\alpha',F) +  M/2^n\geq \pi(\alpha,F)$, and $(\hat H^n,\hat q^n)$ achieves the target expected success ratio $\alpha$, starting from $x_n$.
\end{proof}

Theorem \ref{convergence} applies to trading in options with Lipschitz payoff, e.g. call and put options. The constant $M$ increases with the Lipschitz constants of $F$ and $\phi$, and the value of $|p|$, $d$, $k$, which are inputs of the model. $M$ also increases with the bounds of $|\hat q^n|$ and $||\hat H^n||$ from Lemma \ref{price-Hq-bound}. From the proof of Lemma \ref{H-bound-n}, the bound of $||\hat H^n||$ increases the bounds of $|F|$, $|\hat q^n|$ and the distribution of stock prices, in particular, on the constant $c$ defined in \ref{c}.

Finally, we discuss the approximation using $\prob^{n,0}$. Note that Theorem \ref{convergence} holds as long as NA$_{\phi}(\Omega^n,\prob^{n,\lambda})$ holds, even if $\lambda = 0$. From Example \ref{counter}, this is not true with every convex $\prob$. However, if $\prob^{n,0}$ always allows price change greater than a given (arbitrarily) small threshold, e.g. $\prob$ includes continuous distribution with strictly positive density on $\Omega$ or $\prob$ includes all probability measures on $\Omega$, as considered in \cite{BHZ13,Dolinsky14}, then assigning positive probability to $A^{n,T}$ ($\lambda>0$) is not necessary, for the no arbitrage condition to hold. The next definition formalizes this idea:

\begin{definition}\label{JC}
Given a constant $c>0$, for $0\leq t\leq T-1$ and $(\omega^n_0,\dots\omega^n_t)\in \Omega^n_t$, let 
\begin{align*}
&\mathcal K(c,\omega^n_0,\dots,\omega^n_t) = \left\{K\subset\left(\Omega_1\cap \mathcal D^d_n\right) \cup E_{t+1}: K = \cap_{1\leq i\leq d} K_i\right\},
\end{align*}
where $K_i = \{\omega^n_{t+1}\in\left(\Omega_1\cap \mathcal D^d_n\right) \cup E_{t+1}: \omega^n_{t+1,i} \geq \omega^n_{t,i} + c\}$, or $\left\{\omega^n_{t+1}\in\left(\Omega_1\cap \mathcal D^d_n\right) \cup E_{t+1}: \omega^n_{t+1,i} \leq \omega^n_{t,i} - c\right\}$. 

For each $0\leq t\leq T-1$, the property $\mathcal L_t(c)$ holds for $\left(\Omega^n,\prob^{n,\lambda}\right)$, if for every non-$\prob^{n,\lambda}$-polar $(\omega^n_0,\dots,\omega^n_t)\in\Omega^n_t$, and every $K\in\mathcal K(c,\omega^n_0,\dots,\omega^n_t)$, there exists $P^n\in\prob^{n,\lambda}$, such that $P^n_t(\omega^n_0,\dots,\omega^n_t; K)>0$. The property $\mathcal L(c)$ holds if $\mathcal L_t(c)$ holds for every $0\leq t\leq T-1$.
\end{definition}

Note that with $\lambda >0$ and $c = \min\left\{b-1, 1-a, \bar b_{t+1} -\bar b_t, \underline a_{t}-\underline a_{t+1}, 1\leq t\leq T-1\right\}$, $\mathcal L(c)$ holds for $\left(\Omega^n,\prob^{n,\lambda}\right)$, while with $\lambda = 0$, it is not necessarily true (see e.g. Example \ref{counter}). In addition to $\mathcal L(c)$, we also need to compare $\prob^{n,0}$ and $\prob$, in the following sense:

\begin{definition}\label{Pinclude}
Suppose $\prob \subset \allprob(\Omega)$ and $\prob^n \subset \allprob(\Omega^n)$ and $A^{n,T}$ is $\prob^n$-polar. $\prob$ includes the model $P^n\in\prob^n$, denoted as $P^n\in_n\prob$, if there exists $P\in \prob$, such that for every $B\in \mathcal{B}\left(\Omega^n\setminus A^{n,T}\right)\subset \mathcal B(\Omega)$, $P^n(B) = P(B)$, i.e. $P^n$ can be regarded as a probability on $\mathcal{B}(\Omega)$, which is only supported on $\mathcal{B}\left(\Omega^n\setminus A^{n,T}\right)$. Denote as $\prob^n\subset_n\prob$, if for every $P^n\in\prob^n$, $P^n\in_n\prob$.
\end{definition}

The next theorem verifies NA$_{\phi}(\Omega^n,\prob^{n,0})$ if $\mathcal L(c)$ holds. If in addition $\prob$ includes sufficiently many tree models such that $\prob^{n,0}\subset_n \prob$, then $\pi^n(\alpha,F)$ converges to $\pi(\alpha,F)$. These assumptions hold if the agent solely believes in tree models (see the example in Section \ref{example}), but in general do not exclude continuous models.

\begin{theorem}\label{convergence-c}
If there exists $c>0$, such that $\mathcal L(c)$ holds for $\left(\Omega^n,\prob^{n,0}\right)$, then NA$_{\phi}(\Omega^n,\prob^{n,0})$ holds for $n$ sufficiently large, and Theorem \ref{convergence} holds with $\alpha' = \alpha$. Furthermore, if $\prob^{n,0}\subset\prob$, then $0\leq \pi(\alpha,F)-\pi^n(\alpha,F) \leq M/2^n$ for $M>0$, independent of $n$ and $\alpha$.
\end{theorem}

\begin{proof}
For any $H\in\dstra^n$, since $\mathcal L(c)$ holds for $\left(\Omega^n,\prob^{n,0}\right)$, for every $0\leq t\leq T-1$, and $(\omega^n_0,\dots,\omega^n_t)\in \Omega^n_t$ that is not $\prob^{n,0}$-polar, there exists a non-$\prob^{n,0}$-polar path $(\omega^n_{t+1},\dots,\omega^n_T)$, such that on this path, for any $t\leq s\leq T-1$ and $1\leq i\leq d$, $H_{s,i}(S_{s+1,i}-S_{s,i}) \leq 0$ and $|S_{s+1,i} - S_{s,i}|\geq c$. Thus, similar to the proof for Lemma \ref{NA-dynamic}, if $(H\cdot S)_T \geq 0$ $\prob^{n,0}$-q.s., then $H_t = 0$ $\prob^{n,0}$-q.s. for $0\leq t\leq T-1$ by induction, and NA$(\Omega^n,\prob^{n,0})$ holds. Then the same argument as for Theorem \ref{NA} shows that NA$_{\phi}(\Omega^n,\prob^{n,0})$ holds for $n$ sufficiently large. We skip the details here.

Furthermore, by replacing the constant $\min\left\{b-1, 1-a, \bar b_{t+1} -b_t, \underline a_{t}-\underline a_{t+1}, 1\leq t\leq T-1\right\}$ in the proof of Lemma \ref{H-bound-n} by $c>0$ in the assumption, and considering the path $(\omega^n_{t+1},\dots,\omega^n_T)$ constructed above, the same argument as for Lemma \ref{H-bound-n} implies that $|H_t|$ is bounded for $0\leq t\leq T-1$ by induction. Then Lemma \ref{price-Hq-bound} implies that the optimal quantile hedging strategy $(\hat H^n,\hat q^n)$ for the target expected success ratio $\alpha$ is bounded, uniformly in $n$.

Let $F^{'n} \in \mathcal A^{n,F}(\alpha)$, such that $G^{\pi^n(\alpha,F),\hat H^n,\hat q^n}\geq F^{'n}$ $\prob^{n,0}$-q.s., $C_1$ be the bound for $\hat H^n$ and $\hat q^n$, $C_2$ and $C_3$ be the Lipschitz constants for $\phi$ and $F$, respectively. Then similar to the proof of Theorem \ref{convergence}, we can construct a semi-analytic $F':\Omega \rightarrow \Real_+$, such that $F'(\omega) = F^{'n}(\omega^n)$, if $\omega\in J^n(\omega^n)$, and $G^{\pi^n(\alpha,F)+M/2^n,\hat H^n,\hat q^n} \geq F'$ $\prob$-q.s., with $M = C_1\left(T d+C_2 k\right)+ C_3$.

Furthermore, since $\lambda = 0$ in Definition \ref{ndef}, (\ref{PtoPn}) implies that for every $\omega^n\in\Omega^n\setminus A^{n,T}$, $P^n\in\prob^n$ and the corresponding $P\in\prob$ satisfy $P^n(\omega^n) = P(J^n(\omega^n))$ and $P^n(A^{n,T}) = 0$, where $A^{n,T}$ is defined in (\ref{An}). Thus from (\ref{alphap}) (with $\lambda = 0$), $F'\in\mathcal A(\alpha)$, and that $\pi^n(\alpha,F) +  M/2^n\geq \pi(\alpha,F)$. 

On the other hand, since $\prob^{n,0}\subset_n \prob$, and there is only finitely many paths in $\Omega^n$, as defined in (\ref{barPn}), $\bar P^n \in_n \prob$. Since all martingales in $\QP^n$ are equivalent to $\bar P^n$, $\QP^n \subset_n \QP$. Thus each $P^n\in\prob^{n,0}$ and $Q^n\in\QP^n$ (with a slight abuse of notation) can also regarded as a probability measure on $\borel(\Omega)$, but only supported on $\borel(\Omega^n\setminus A^{n,T})$.

Consider maximizing the expected success ratio with initial capital $x\geq 0$. If $G\in \mathcal C(x)$, then define $G^n:\Omega^n \rightarrow \Real_+$ by $G^n = G1_{\{\Omega^n\setminus A^{n,T}\}}$. Since $\QP^n\subset_n \QP$, and $A^{n,T}$ is $\QP^n$-polar,
\begin{equation*}
\pi^n\left(G^n\right) = \sup\limits_{Q^n\in\QP^n}\mathbb{E}_{Q^n}\left[G^n\right] =\sup\limits_{Q^n\in\QP^n}\mathbb{E}_{Q^n}\left[G\right]\leq \sup\limits_{Q\in\QP}\mathbb{E}_{Q}\left[G\right]  \leq x.
\end{equation*}
Thus $G^n \in \mathcal C^n(x)$. Furthermore, since $\psi^{G^n} = \psi^{G}$ on $\Omega^n\setminus A^{n,T}$, $\prob^{n,0}\subset_n\prob$, and $A^{n,T}$ is $\prob^{n,0}$-polar, we obtain $\displaystyle\inf_{P^n\in\prob^{n,0}}\mathbb{E}_{P^n}\left[\psi^{G^n}\right] =   \displaystyle\inf_{P^n\in\prob^{n,0}}\mathbb{E}_{P^n}\left[\psi^{G}\right] \geq \displaystyle\inf_{P\in\prob}\mathbb{E}_P\left[\psi^{G}\right]$. Thus, 
\begin{equation*}
V(x,F)=\sup_{G \in \mathcal{C}(x)}\inf_{P\in\prob} \mathbb{E}_P\left[\psi^{G}\right] \leq \sup_{G \in \mathcal{C}(x)}\inf_{P^n\in\prob^{n,0}} \mathbb{E}_P\left[\psi^{G}\right]\leq \sup_{G^n \in \mathcal{C}^n(x)}\inf_{P^n\in\prob^{n,0}} \mathbb{E}_{P^n}\left[\psi^{G^n}\right] = V^n(x,F).
\end{equation*}
Thus, from the proof of Proposition \ref{inverse},
\begin{equation*}
\pi(\alpha,F) \geq \inf\{x \geq 0: V(x,F) \geq \alpha\} \geq \inf\left\{x\geq 0: V^n(x,F) \geq \alpha\right\} = \pi^n(\alpha,F). \qedhere
\end{equation*}
\end{proof}

\section{\textbf{Examples}}\label{example}
This section shows two examples in which $\prob$ does not have a dominating measure and Theorem \ref{convergence-c} applies so that the quantile hedging price can be approximated from the discretized market. In the first example, $\prob$ includes all finitely supported tree models, so that every path can be of probability one under some model, and the quantile hedging price is the superhedging price multiplied by the target expected success ratio $\alpha$. The second examples assumes a smaller $\prob$, and the quantile hedging price is convex in $\alpha$. We set $T=d=1$ in both examples, for ease of calculation. Multi-period models add more constraints to $P^n\in\prob^{n,0}$ and $Q^n\in\QP^n$, but do not change the nature of the nonlinear programming problem.

\textbf{Example 1}. Suppose there is one stock with $S_0 = 1$ and $\Omega_1 = [1/2,3/2]$. Let $\prob$ be the set of all probability measures on $\Omega_1$ that is supported on finitely many paths. Since $S_1 - S_0$ can be both positive and negative with positive probability under some model in $\prob$, the market satisfies NA$(\Omega,\prob)$. Since there is a continuum of paths that are not $\prob$-polar, $\prob$ is not dominated. By the construction in Definition \ref{ndef}, $\prob^{n,0} \subset_n \prob$, and for $n\geq 1$, $\mathcal L(0.25)$ holds for $\left(\Omega^n,\prob^{n,0}\right)$. Thus Theorem \ref{convergence-c} applies.

Given a set of options $\phi$ which does not create arbitrage opportunities, the maximum expected success ratio with any contingent claim $F$, and initial capital $x\geq 0$ in the discretized market is
\begin{align}
V^n(x,F)=\inf_{a\geq 0}\left\{xa + \inf_{\prob^{n,0}\times\QP^n}\sum\limits_{w^n\in \Omega^n}\left(P^n(w^n)-a Q^n(w^n)F(w^n)\right)^{+}\right\},\label{mini}
\end{align}
 of which the inverse is the quantile hedging price.
\begin{proposition}\label{linear}
With $\prob$ defined above, in the discretized market with model uncertainty $\prob^{n,0}$, $\pi^n(\alpha,F) = \alpha\pi^n(F)$.
\end{proposition}

\begin{proof}
For any $Q^n\in\QP^n$, let $s(Q^n) = \sum\limits_{\omega^n\in\Omega^n}Q^n(\omega^n)F(\omega^n)$ and $s^*(n) = \displaystyle\sup_{Q^n\in\QP^n}s(Q^n)$. In the minimization problem in (\ref{mini}), given the constant $a\geq 0$, there are three cases:

(i) If $s^*(n)\geq \frac{1}{a}$ and there exists $\omega^n\in\Omega^n$ and $Q^n\in\QP^n$, such that $aQ^n(\omega^n)F(\omega^n) \geq 1$, then since $Q^n$ is only supported on $\Omega^n\setminus A^{n,T}$, there exists a sequence of $\left\{P^{n,i}\right\}_{i\geq 1} \subset \prob^{n,0}$, such that for each $i\geq 1$, $P^{n,i}$ is equivalent to $Q^n$, and $P^{n,i}(\omega^n)$ increases to $1$. Thus $\inf\limits_{\prob^{n,0}\times\Q^n}\sum\limits_{\omega^n\in\Omega^n}(P^n(\omega^n)-aQ^n(\omega^n)F(\omega^n))^{+} = 0$.

(ii) If  $s^*(n)\geq\frac{1}{a}$ and for every $\omega^n\in\Omega^n$ and $Q^n\in\QP^n$, $aQ^n(\omega^n)F(\omega^n) < 1$, then there exists a sequence $\left\{Q^{n,i}\right\}\subset\QP^n$, such that $s(Q^{n,i})$ increases and converges to $\frac{1}{a}$. Define $P^{n,i}\in\prob^{n,0}$ by $P^{n,i}(\omega^n) = \frac{Q^{n,i}(\omega^n)F(\omega^n)}{s(Q^{n,i})}\geq aQ^{n,i}(\omega^n)F(\omega^n)$. Then $\sum\limits_{\omega^n\in\Omega^n}\left[(P^{n,i}(\omega^n)-aQ^{n,i}(\omega^n)F(\omega^n))^{+}\right] = 1 - a s(Q^{n,i})$, which converges $0$. Thus $\inf\limits_{\prob^{n,0}\times\Q^n}\sum\limits_{\omega^n\in\Omega^n}\left(P^n(\omega^n)-aQ^n(\omega^n)F(\omega^n)\right)^{+} = 0$.

(iii) If  $s^*(n) < \frac{1}{a}$. For any $Q^n\in\QP^n$, if $s(Q^n) = 0$, then since $F\geq 0$, $Q^n(\omega^n)F(\omega^n) = 0$ for every $\omega^n\in\Omega^n$, and  $\inf\limits_{\prob^{n,0}}\sum\limits_{\omega^n\in\Omega^n}(P^n(\omega^n)-aQ^n(\omega^n)F(\omega^n))^{+} = 1$. 

Otherwise, define $\hat P^n\in\prob^{n,0}$ by $\hat P^n(\omega^n) = Q^n(\omega^n)F(\omega^n)/s(Q^n) \geq aQ^n(\omega^n)F(\omega^n)$ for every $\omega^n\in\Omega^n$. Then for any $\tilde P \in \prob^{n,0}$, 
\begin{align*}
&\sum\limits_{\omega^n\in\Omega^n}\left(\tilde P^n(\omega^n)-aQ^n(\omega^n)F(\omega^n)\right)^{+}\geq \sum\limits_{\omega^n\in \Omega^n}\left(\tilde P^n(\omega^n)-aQ^n(\omega^n)F(\omega^n)\right)\\
=&1-as(Q^n) = \sum\limits_{\omega^n\in \Omega^n}\left(\hat P^n(\omega^n)-aQ^n(\omega^n)F(\omega^n)\right)=\sum\limits_{\omega^n\in \Omega^n}\left(\hat P^n(\omega^n)-aQ^n(\omega^n)F(\omega^n)\right)^+.
\end{align*}
Thus, $\displaystyle\inf_{\prob^{n,0}}\sum\limits_{\omega^n\in\Omega^n}(P^n(\omega^n)-aQ^n(\omega^n)F(\omega^n))^{+} = \sum\limits_{\omega^n\in \Omega^n}\left(\hat P^n(\omega^n)-aQ^n(\omega^n)F(\omega^n)\right)^+=  1-as(Q^n)$. Therefore  $\inf\limits_{\prob^{n,0}\times\QP^n}\sum\limits_{\omega^n\in\Omega^n}(P^n(\omega^n)-aQ^n(\omega^n)F(\omega^n))^{+} = 1-as^*(n)$.

To summarize the three cases above, for $x\geq 0$,
\begin{align*}
&xa + \inf_{\prob^{n,0}\times\QP^n}\sum\limits_{\omega^n\in\Omega^n}^n(P^n(\omega^n)-aQ^n(\omega^n)F(\omega^n))^{+}
=\begin{cases} xa &\text{ if } a \geq 1/s^{*}(n)\\
xa + (1-as^*(n)) &\text{ if } a < 1/s^{*}(n),
\end{cases}
\end{align*}
and $V^n(x,F) = \inf\limits_{a\geq 0}\left\{xa + \inf\limits_{\prob^{n,0}\times\QP^n}\sum\limits_{\omega^n\in\Omega^n}^n(P^n(\omega^n)-aQ^n(\omega^n)F(\omega^n))^{+}\right\}
= \begin{cases} \frac{x}{s^*(n)} &\text{ if } x < s^{*}(n)\\
1 &\text{ if } x \geq s^{*}(n).
\end{cases}$

The quantile hedging price is the inverse function of $V^n(x)$: $\pi^n(\alpha,F) = \alpha s^*(n)$ for $0\leq \alpha \leq 1$, which indicates that $s^*(n)$ is the superhedging price.
\end{proof}

The intuition behind Proposition \ref{linear} is that any $\omega^n$  is of probability arbitrarily close to $1$ under some model in $\prob^{n,0}$, so that to quantile hedge $F$ with expected success ratio $\alpha<1$, the hedging strategy needs to superhedge $\alpha F$ on $\omega^n$. Since Theorem \ref{convergence-c} applies, $\pi(\alpha,F) = \alpha\pi(F)$ also holds. This example also demonstrates the difference between quantile hedging with expected success ratio and quantile hedging with success probability, because for the latter, quantile hedging with any given probability is actually superhedging.

\textbf{Example 2}. Let $\bar{\prob}$ includes all convex combination of finitely many $P$'s from $\prob$ in Example 1, which also satisfies the following properties:

(i) $P\leq 0.05$,

(ii) For $n = 5$ and $6$, and each $\omega'\in \{1\}\times \D_n\cap \left[\frac 1 2, \frac 3 2\right]$, there is at most one $\omega\in J^n(\omega')$, such that $P(\omega)>0$, 

(iii) $P\left(\left[\frac 5 4,\frac 3 2\right]\right)\geq 0.25$, and $P\left(\left[\frac 1 2,\frac 3 4\right]\right)\geq 0.25$.

In other words, $\bar\prob$ includes all convex combinations of models from $\prob$, with restriction that each model (i) assigns probability of no more than 0.05 to any path, (ii) assigns positive probability to only one path in every (narrow) interval $J^5(\omega')$ or $J^6(\omega')$, and (iii) the probability of both the first and the fourth quarter of the range of $S_1$ is greater than or equal to 0.25. Denote set of probability measures on $\borel(\Omega^n)$ in Definition \ref{ndef} as $\bar \prob^{n,\lambda}$. 

\begin{lemma}
For $n\geq 5$, $\bar\prob^{n,0} \subset_n \bar\prob$ and $\mathcal L(0.25)$ holds for $\left(\Omega^n,\prob^{n,0}\right)$. Thus Theorem \ref{convergence-c} applies.
\end{lemma}

\begin{proof}
For each $P\in\bar \prob$, suppose $P$ is a convex combination of $P^i$,  $1\leq i\leq m \in \mathbb N$, and each $P^i$ satisfies properties (i)-(iii) in the definition of $\bar \prob$ above. Let $P^n, P^{n,i}$ be the discretized probability measures in $\bar\prob^{n,0}$ corresponding to $P,P^i$, $1\leq i\leq m$, respectively. Then $P^n$ is a convex combination of $P^{n,i}$, $1\leq i\leq m$. Furthermore, each $P^{n,i}$, if regarded as a probability measure on $\borel (\Omega)$, inherits from $P^i$ properties (i)-(iii). Therefore $P^n\in_n \bar\prob$ and $\bar\prob^{n,0}\subset_n\bar{\prob}$.

Finally, each $P^{n,i}$ satisfying (iii) implies that $P^n\left(\left[\frac 5 4,\frac 3 2\right]\right)\geq 0.25$, and $P^n\left(\left[\frac 1 2,\frac 3 4\right]\right)\geq 0.25$. Thus $\mathcal L(0.25)$ holds for $\left(\Omega^n,\prob^{n,0}\right)$.
\end{proof}

Figure \ref{price1} shows the quantile hedging price in the discretized market for $n = 5,6,7,8,9$, for $F = S_1^2$, without trading in options\footnote{The minimization in (\ref{mini}) is carried out by using CVX, a package for specifying and solving convex programs \cite{cvx,gb08}. We thank Agostino Capponi for recommending CVX to us.}. The quantile hedging prices converge quickly. $\pi^8(\alpha,F)$ (with 257 paths) and $\pi^9(\alpha,F)$ (with 513 paths) almost coincide, and can be used as the approximation of the quantile hedging price in the original market.

\begin{figure}
\centering
\includegraphics[width=0.8\textwidth]{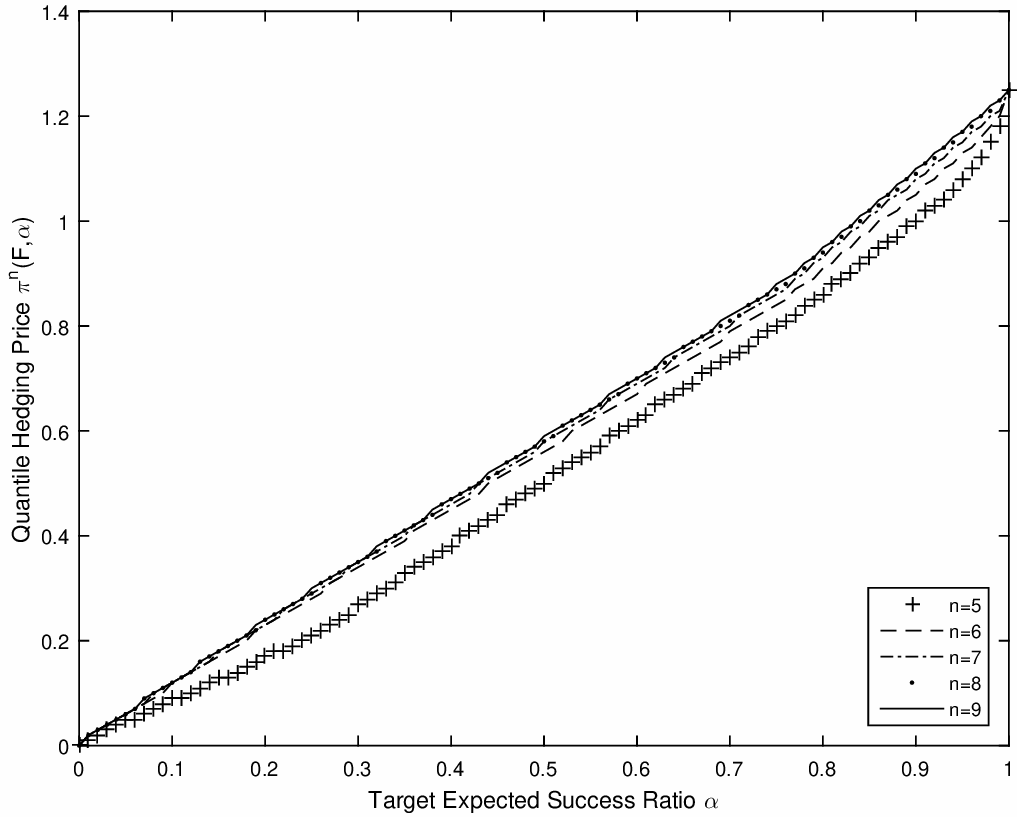}
\caption{The quantile hedging price for $F = S_1^2$ in the discretized market with model uncertainty $\bar\prob^{n,0}$, for $n = 5,6,7,8,9$.}
\label{price1}
\end{figure}

There are less models in $\bar \prob$ than in $\prob$ and the quantile hedging price shows some convexity in the target expected success ratio $\alpha$, which agrees with Proposition \ref{inverse-n}. The left panel of Figure \ref{price2} shows the difference between the quantile hedging price with model uncertainty $\prob^{9,0}$ and $\bar \prob^{9,0}$. As Proposition \ref{linear} shows, the solid line corresponding to $\prob^{9,0}$ is a straight line, while the dotted curve corresponding to $\bar{\prob^{9,0}}$ is convex in $\alpha$.

Finally, we add two put options on the stock, with strike prices 0.75 and 1, and market prices $0.075$ and $0.2$, respectively. The right panel of Figure \ref{price2} shows the quantile hedging price with and without trading in these options, under model uncertainty $\bar\prob^{9,0}$. The quantile hedging price is always cheaper with trading in options.

\begin{figure}
\includegraphics[width=.45\textwidth]{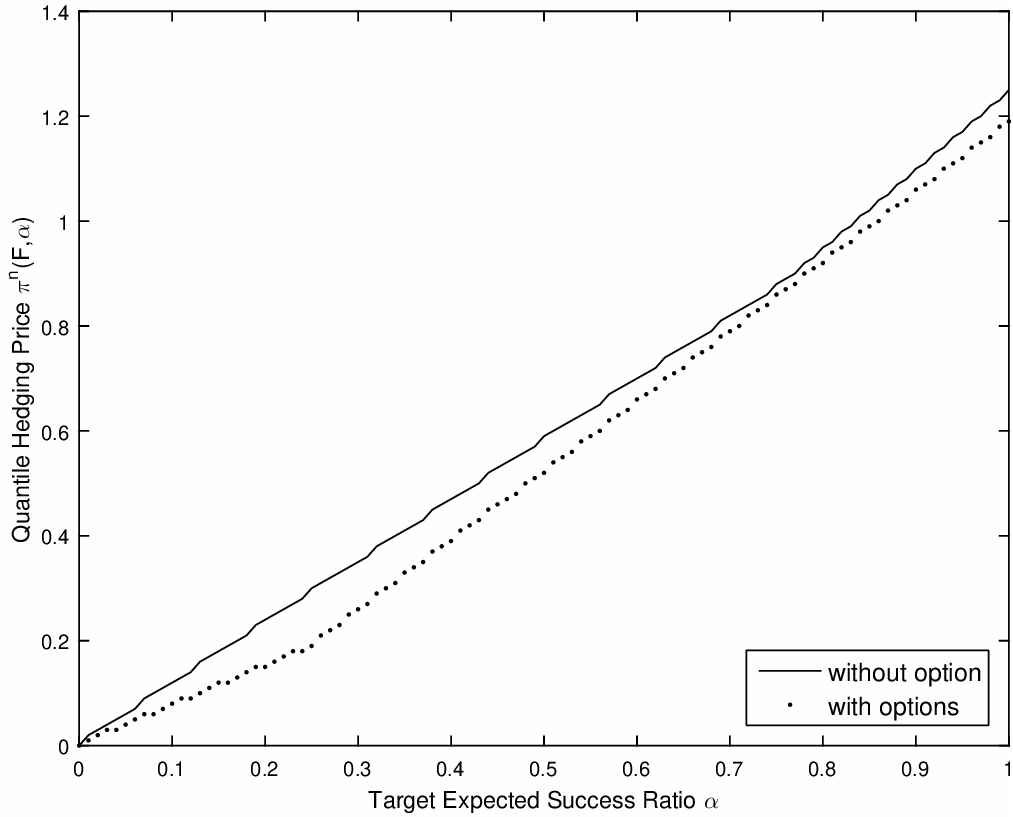}
~\hspace{0.15in}
\includegraphics[width=.45\textwidth]{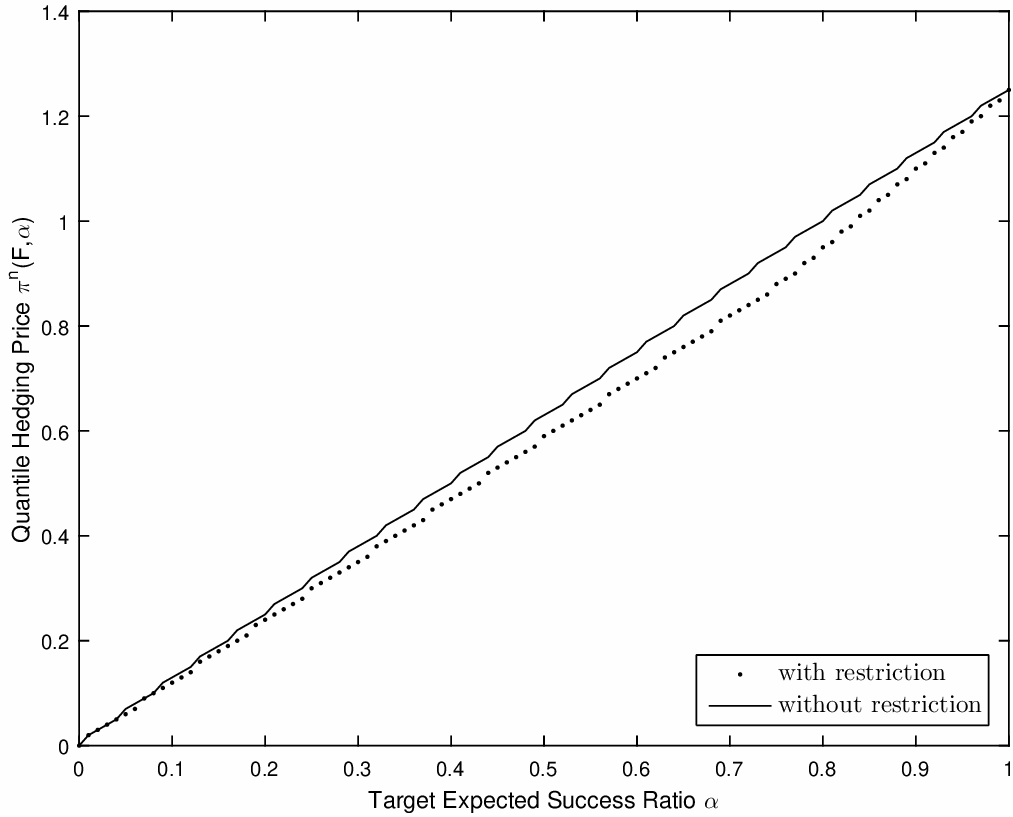}
\caption{Left panel: the comparison between the quantile hedging price for $F=S_1^2$, with (under $\bar\prob^{9,0}$, dotted line) and without (under $\prob^{9,0}$, solid line) restriction on the model uncertainty. Right panel: the comparison between the quantile hedging price with (dotted line) and without (solid line) trading two put options, under model uncertainty $\bar\prob^{9,0}$.}
\label{price2}
\end{figure}

\appendix
\numberwithin{lemma}{section}
\numberwithin{proposition}{section}
\section*{Appendix}
\setcounter{lemma}{0}
\setcounter{proposition}{0}
\renewcommand{\thelemma}{A.\arabic{lemma}}
\renewcommand{\theproposition}{A.\arabic{proposition}}
\renewcommand{\theequation}{A.\arabic{equation}}

\begin{lemma}\label{NA-dynamic}
If $\lambda>0$, then NA$(\Omega^n,\prob^{n,\lambda})$ holds.
\end{lemma}

\begin{proof}
We prove this lemma by contradiction: suppose there exits an arbitrage strategy $H\in \dstra^n$, i.e. $(H\cdot S)_T \geq 0$ $\prob^n$-q.s., and we show by induction that $(H\cdot S)_T = 0$ $P^n$-a.s. under every $P^n\in\prob^{n,\lambda}$. 

Fix a $P^n\in\prob^n$, For $t\geq 1$, let $\omega^n_{t,i} = \bar b_t 1_{\{H_{t-1,i}(\omega^n_0,\dots,\omega^n_{t-1})<0\}} + \underline a_t 1_{\{H_{t-1,i}(\omega^n_0,\dots,\omega^n_{t-1})\geq 0\}}$, where the second index $i$ in the subscripts of $d$-vector $\omega^n_t$ and $H_t$ indicates their $i$-th entry, for any $0\leq t\leq T$. 

Consider the path $\omega^n = \left(\omega^n_1,\dots,\omega^n_T\right)$. Since $\lambda>0$, $P^n(\omega^n) >0$, and therefore $(H\cdot S)_T(\omega^n) \geq 0$. On the other hand, letting 
\begin{equation}
c = \min\left\{b-1, 1-a, \bar b_{t+1} -\bar b_t, \underline a_{t}-\underline a_{t+1}, 1\leq t\leq T-1\right\}>0, \label{c}
\end{equation}
then on $\omega^n$, for any $0\leq s\leq T-1$ and $1\leq i\leq d$, $H_{s,i}(S_{s+1,i}-S_{s,i}) \leq 0$, and $|S_{s+1,i} - S_{s,i}|\geq c$, which implies that for each $1\leq i\leq d$, $H_{0,i}(S_{1,i}-S_{0,i}) \geq -\sum\limits_{s=1}^{T-1}H_s(S_{s+1}-S_s) - \sum\limits_{j\neq i}^d H_{0,j}(S_{1,j}-S_{0,j}) \geq 0$. Thus $H_{0} =\mathbf 0$ and therefore $(H\cdot S)_1 = 0$.

For $1\leq t\leq T-1$, suppose $(H\cdot S)_{t} = 0$ $P^n$-a.s. For every $(\omega^n_0,\dots,\omega^n_t)\in\Omega^n_t$, similar to the argument for $H_0$, there exists $(\omega^n_{t+1},\dots,\omega^n_T)$ with positive probability, on which for any $t\leq s\leq T-1$ and $1\leq i\leq d$, $H_{s,i}(S_{s+1,i}-S_{s,i}) \leq 0$ and $|S_{s+1,i} - S_{s,i}|\geq c$. Then $(H\cdot S)_T \geq 0$ $P^n$-a.s. implies that $H_{t}(\omega_0,\dots,\omega_t) = \mathbf 0$ for $P^n$-a.s. every $(\omega^n_0,\dots,\omega^n_t)\in \Omega^n_t$. Thus $(H\cdot S)_{t+1} = 0$ $P^n$-a.s.
\end{proof}

\begin{lemma}\label{H-bound-n}
If $\lambda >0$, $H\in\dstra^n$ and $(H\cdot S)_T \geq -D$ $\prob^{n,\lambda}$-q.s., where $D >0$ is independent of $n$, then there exists $\tilde H\in\dstra^n$ such that $(\tilde H\cdot S)_T = (H\cdot S)_T$ $\prob^{n,\lambda}$-q.s. and $||\tilde H||<M$ for some $M>0$ independent of $n$.
\end{lemma}

\begin{proof}
For $0\leq t\leq T-1$, if   $(\omega^n_0,\dots,\omega^n_t)\in\Omega^n_t$ is $\prob^{n,\lambda}$-polar, then let $\tilde H_t(\omega^n_0,\dots,\omega^n_t) = 0$. Otherwise, let $\tilde H_t = H_t$. Thus $(\tilde H\cdot S)_T = (H\cdot S)_T$ $\prob^{n,\lambda}$-q.s. We show the boundedness of $\tilde H$ by induction.

As argued in the proof of Lemma \ref{NA-dynamic}, for every $0\leq t\leq T-1$ and non-$\prob^{n,\lambda}$-polar $(\omega^n_0,\dots,\omega^n_t)\in \Omega^n_t$, there exists a non-$\prob^{n,\lambda}$-polar path $(\omega^n_{t+1},\dots,\omega^n_T)$, on which for any $t\leq s\leq T-1$ and $1\leq i\leq d$, $H_{s,i}(S_{s+1,i}-S_{s,i}) \leq 0$ and $|S_{s+1,i} - S_{s,i}|\geq c$, where $c$ is defined in (\ref{c}). 

Thus for $\tilde H_0$, since $(\tilde H\cdot S)_T \geq -D$ $\prob^n$-q.s., by on the path constructed above, for every $1\leq i\leq d$, $H_{0,i}(S_{1,i} - S_{0,i}) \geq -\sum\limits_{s=1}^{T-1}H_s(S_{s+1}-S_s) - \sum\limits_{j\neq i}^d H_{0,j}(S_{1,j}-S_{0,j}) -D \geq -D$, and thus $ |H_{0,i}|\leq \frac{D}{c}$.
		
For $1\leq t\leq T-1$, assume that for $0\leq s \leq t-1 $, $|\tilde H_s|$ are bounded. For every non-$\prob^{n,\lambda}$-polar $(\omega^n_0,\dots,\omega^n_t) \in \Omega^n_t$, consider the path constructed above. Since $(\tilde H\cdot S)_T \geq -D$ $\prob^{n,\lambda}$-q.s., and $\tilde H_s$ are bounded for $0\leq s\leq t-1$, on this non-$\prob^{n,\lambda}$-polar path, $\sum\limits_{s=t}^{T-1}\tilde H_s(\omega^n_0,\dots,\omega^n_s)(\omega^n_{s+1}-\omega^n_s) > -D_t$, where $D_t >0$ is independent of $n$. Then similar to the argument for $H_0$, for every $1\leq i \leq d$, $H_{t,i}(S_{t+1,i}-S_{t,i}) \geq -D_t$, and thus $|\tilde H_{t,i}| \leq D_t/ c$, which is independent of $n$ and $i$.
\end{proof}

\bibliographystyle{siam}
\bibliography{hedge}
\end{document}